\documentclass[a4paper,11pt]{article}
\usepackage[latin1]{inputenc}
\usepackage[T1]{fontenc}
\usepackage{amsmath,amssymb}
\usepackage{array,booktabs}
\numberwithin{equation}{section}
\usepackage[toc,page]{appendix}
\usepackage{amsthm}
\usepackage{hyperref}
\newtheorem{thm}{Theorem}
\newtheorem{lemma}[thm]{Lemma}
\newtheorem{cor}[thm]{Corollary}
\theoremstyle{definition}
\numberwithin{thm}{section}
\usepackage{verbatim}
\def\R{\ensuremath\mathbb{R}}
\def\E{\ensuremath\mathbb{E}}
\def\Z{\ensuremath\mathbb{Z}}

\def\N{\ensuremath\mathbb{N}}
\def\eps{\ensuremath\varepsilon}
\newcommand{\balls}{m}
\newcommand{\bins}{n}
\newcommand{\trm}{\mu_0} %Truly random mean number of full bins
\newcommand{\trp}{p_0}
\newcommand{\fbins}{|h(X)|}

\newcommand{\OO}{O}

\newcommand{\Olog}{\tilde{O}}
\newcommand{\var}{\mathop{\mathrm{Var}}}
\usepackage{fullpage}

\author{Anders Aamand and Mikkel Thorup}
\title{Non-Empty Bins with Simple Tabulation Hashing}
\begin{document}
\maketitle
\begin{abstract}

%We consider the hashing of a set $X\subseteq U$ with $|X|=\balls$ using a simple tabulation hash function $h:U\to [\bins]$ and analyse the number of non-empty bins~i.e. the size of $h(X)$. We estimate the expected size of $h(X)$ which turns out to  match that of the fully random setting asymptotically and we prove that $\fbins$ is somewhat concentrated around its mean. As the number of non-empty bins is a fundamental measure in the balls and bins paradigm applications are numerous. We will discuss two, namely Bloom filters and Filter hashing.

We consider the hashing of a set $X\subseteq U$ with $|X|=\balls$ using
a simple tabulation hash function $h:U\to [\bins]=\{0,\dots,\bins-1\}$ and analyse the
number of non-empty bins, that is, the size of $h(X)$. We show that the
expected size of $h(X)$ matches that with fully random hashing to within
low-order terms. We also provide concentration bounds.  The number of
non-empty bins is a fundamental measure in the balls and bins
paradigm, and it is critical in applications such as Bloom filters and
Filter hashing.  For example, normally Bloom filters are proportioned for a desired low
false-positive probability assuming fully random hashing (see \url{en.wikipedia.org/wiki/Bloom_filter}).
Our results imply that if we implement the hashing with simple tabulation, we obtain the same
low false-positive probability for any possible input.
\end{abstract}

\thispagestyle{empty}
\setcounter{page}0
\newpage\section{Introduction}
%In this paper we consider the hashing of a set $X\subseteq U$ with $|X|=\balls$ using a simple tabulation hash function $h:U\to [\bins]$. We will study the number of non-empty bins i.e.~the size of $h(X)$ which we believe to be a natural measure worth studying for hash functions in general. 

We consider the balls and bins paradigm where a set $X\subseteq U$ of
$|X|=\balls$ balls are distributed into a set of $\bins$ bins
according to a hash function $h:U \to [\bins]$. We are interested in questions relating to the distribution of $\fbins$, for example:
What is the expected number of non-empty bins? How well is $\fbins$ concentrated around its mean? And what is
the probability that a query ball lands in an empty bin?
These questions are critical in applications such as 
Bloom filters~\cite{Bloom} and Filter hashing~\cite{Fot}.

In the setting where $h$ is a fully random hash function, meaning that
the random variables $(h(x))_{x\in U}$ are mutually independent and
uniformly distributed in $[\bins]$, the situation is well
understood. The random distribution process is equivalent to throwing
$\balls$ balls sequentially into $\bins$ bins by for each ball
choosing a bin uniformly at random and independently of the placements
of the previous balls. The probability that a bin becomes empty is
thus $\left(1-1/\bins \right)^\balls$; so the expected number of
non-empty bins is exactly $\trm:=\bins (1-(1-1/\bins)^\balls)$ and,
unsurprisingly, the number of non-empty bins turns out to be sharply
concentrated around $\trm$ (see for example Kamath et al.~\cite{Kam}
for several such concentration results).

In practical applications fully random hashing is unrealistic and so
it is desirable to replace the fully random hash functions with
realistic and implementable hash functions that still provide at least
some of the probabilistic guarantees that were available in the fully
random setting. However, as the mutual independence of the keys is
often a key ingredient in proving results in the fully random setting
most of these proofs do not carry over. Often the results are simply no longer true and if
they are one has to come up with alternative techniques for proving
them.

In this paper, we study the number of non-empty bins when the hash
function $h$ is chosen to be a simple tabulation hash function
\cite{Pat,Zobrist}; which is very fast and easy to implement (see description
below in Section~\ref{sec:simple-tab}). We
provide estimates on the expected size of $\fbins$ which
asymptotically match\footnote{Here we use ``asymptotically'' in the
  classic mathematical sense to mean equal to within low order terms,
  not just within a constant factor.}  those with fully random hashing
on any possible input.  To get a similar match within the classic
$k$-independence paradigm~\cite{WC}, we would generally need
$k=\Omega((\log n)/(\log\log n))$. For comparison, simple tabulation is the
fastest known 3-independent hash function~\cite{Pat}. We will 
also study how $\fbins$ is concentrated around its mean.

Our results complements those from~\cite{Pat}, which show that with
simple tabulation hashing, we get Chernoff-type concentration on the
number of balls in a given bin when $\balls\gg \bins$.  For example,
the results from~\cite{Pat} imply that all bins are non-empty with high probability (whp) when $\balls=\omega(\bins\log\bins)$. More precisely, for any
 constant $\gamma>0$, there exists a $C>0$ such that if $\balls \geq C \bins \log \bins$,
  all bins are non-empty with probability  $1-\OO(\bins^{-\gamma})$.
 As a consequence, we only have to study
$\fbins$ for $\balls=O(\bins\log\bins)$ below. On the other hand,~\cite{Pat}
does not provide any good bounds on the probability that a bin is
non-empty when, say, $\balls=\bins$. In this case, our results imply
that a bin is non-empty with probability $1-1/e\pm o(1)$, as in the
fully random case. The understanding we provide here is critical to
applications such as Bloom filters~\cite{Bloom} and Filter
hashing~\cite{Fot}, which we describe in section~\ref{Bloomsec} and~\ref{filtsec}.

We want to emphasize the advantage of having many complementary
results for simple tabulation hashing. An obvious advantage is that
simple tabulation can be reused in many contexts, but there may also be applications
that need several strong properties to work in tandem. If, for example, an application
has to hash a mix of a few heavy balls and many light balls, and the
hash function do not know which is which, then the results from
\cite{Pat} give us the Chernoff-style concentration of the number of
light balls in a bin while the results of this paper give us the right probability
that a bin contains a heavy ball. For another example where an interplay of properties becomes important see section~\ref{filtsec} on Filter hashing.
The reader is referred to~\cite{Tho17} for
a survey of results known for simple tabulation hashing, as well as examples where
simple tabulation does not suffice and where slower more sophisticated hash functions
are needed.

%In this paper we will be focused on the number of non-empty bins when $\balls$ balls are distributed according to a hash function $h$ into $\bins$ bins using a hash function $h$.   

%Consider the hashing of a set $X\subseteq U$ with $|X|=\balls$ to a set of $\bins$ bins using a hash function $h:U\to [\bins]$. Our interest in this paper will be the number of non-empty bins i.e.~the size of $h(X)$ when $h$ is chosen to be a simple tabulation hash function.

%The case where $h$ is fully random, meaning that the random variables $(h(x))_{x\in U}$ are mutually independent and uniformly distributed in $[\bins]$, is one of the classic balls and bins problems and has been studied extensively in the literature (REFERENCES?). Here the random distribution process is equivalent to throwing $\balls$ balls sequentially into $\bins$ bins by for each ball choosing a bin uniformly at random and independently of the placements of the previous balls. The probability that a bin is empty is thus  $\left(1-1/\bins \right)^\balls$ so the expected number of non-empty bins is exactly $\trm:=\bins (1-(1-1/\bins)^\balls)$ and  unsurprisingly the number of non-empty bins turns out to be sharply concentrated around $\trm$ (see for example Kamath et al.~\cite{Kam} for several such concentration results). 

\subsection{Simple tabulation hashing}\label{sec:simple-tab}
Recall that a hash function $h$ is a map from a universe $U$ to a range $R$ chosen with respect to some probability distribution on the set of all such functions. If the distribution is uniform (equivalently the random variables $(h(x))_{x\in U}$ are mutually independent and uniformly distributed in $R$) we will say that $h$ is fully random.

Simple tabulation was introduced by Zobrist~\cite{Zobrist}. For simple tabulation $U=[u]=\{0,\dots,u-1\}$ and $R=[2^r]$ for some $r\in \N$. The keys $x\in U$ are viewed as vectors $x=(x[0],\dots,x[c-1])$ of $c=\OO(1)$ characters with each $x[i]\in \Sigma:=[u^{1/c}]$. The simple tabulation hash function $h$ is defined by
\begin{align*}
h(x)=\bigoplus_{i\in [c]}h_i(x[i]),
\end{align*}
where $h_0,\dots,h_{c-1}:\Sigma \to R$ are independent fully random hash functions and where $\oplus$
denotes the bitwise XOR. What makes it fast is that the character
domains of $h_0,\dots,h_{c-1}$ are so small that they can be stored as
tables in fast cache. Experiments in~\cite{Pat} found that the
hashing of $32$-bit keys divided into $4$ $8$-bit characters was as
fast as two $64$-bit multiplications. Note that on machines with
larger cache, it may be faster to use $16$-bit characters.
As useful computations normally involve data and hence cache, there is no commercial drive for developing processors that do multiplications much faster than cache look-ups. Therefore, on real-world processors, we always expect cache based simple tabulation to be at least comparable in speed to multiplication. The converse is not true, since many useful computations do not involve multiplications. Thus there is a drive to 
make cache faster even if it is too hard/expensive to speed up multiplication circuits.

Other important properties include that the $c$ character table
lookups can be done in parallel and that when initialised the
character tables are not changed. For applications such as Bloom
filters where more than one hash function is needed another nice
property of simple tabulation is that the output bits are mutually
independent. Using $(kr)$-bit hash values is thus equivalent to using
$k$ independent simple tabulation hash functions each with values in
$[2^r]$. This means that we can get $k$ independent $r$-bit hash
values using only $c$ lookups of $(kr)$-bit strings.

\subsection{Main Results}
We will now present our results on the number of non-empty bins with simple tabulation hashing. 
\paragraph{The expected number of non-empty bins:}
Our first theorem compares the expected number of non-empty bins when
using simple tabulation to that in the fully random setting. We denote
by $\trp =1-\left(1-1/\bins \right)^{\balls}<\balls/\bins$ the probability that a
bin becomes non-empty and by $\trm=\bins \trp$ the expected number of
non-empty bins when $\balls$ balls are distributed into $\bins$ bins
using fully random hashing.

\begin{thm}\label{expthm}
Let $X\subseteq U$ be a fixed set of $|X|=\balls$ balls.  Let $y\in
[\bins]$ be any bin and suppose that $h:U\to [\bins]$ is a simple tabulation hash function. If $p$ denotes the probability that
$y\in h(X)$ then
\begin{align*}
|p-\trp|\leq \frac{\balls^{2-1/c}}{\bins^2} \quad \text{and hence} \quad \left|\E[\fbins]-\trm\right| \leq \frac{\balls^{2-1/c}}{\bins}.
\end{align*}
If we let $y$ depend on the hash of a distinguished query ball $q\in
U\backslash X$, e.g., $y=h(q)$, then the bound on $p$ above is replaced by the
weaker $|p-\trp|\leq \frac{2\balls^{2-1/c}}{\bins^2}$.

%If $y$ is chosen dependently on the hash value $h(q)$ of a query ball $q\in U\backslash X$ (e.g.~$y=h(q)$) the bound on $p$ above can be replaced by $|p-\trp|\leq \frac{2\balls^{2-1/c}}{\bins^2}$.
\end{thm}
The last statement of the theorem is important in the application to Bloom filters where we  wish to upper bound the probability that $h(q)\in h(X)$ for a query ball $q\notin X$.
%For the application to Bloom filters it is important that the probability that a given bin is non-empty is approximately that of the fully random setting even when the bin is chosen as (a function of) the hash value of a given query ball $q$. 

To show that the expected relative error
$\left|\E[\fbins]-\trm\right|/\trm$ is always small, we have to
complement Theorem~\ref{expthm} with the result from~\cite{Pat} that
all bins are full, whp, when $\balls\geq~C\bins\log\bins$ for some
large enough constant $C$. In particular, this implies 
$\left|\E[\fbins]-\trm\right|/\trm\leq 1/\bins$ when $\balls\geq C\bins\log\bins$.
The relative error from Theorem~\ref{expthm} is maximized when $\balls$ is
maximized, and with $\balls=C\bins\log\bins$, it
is bounded by $\frac{\balls^{2-1/c}}{\bins^2}=O((\log^2\bins)/\bins^{1/c})=\Olog (\bins^{-1/c})$. Thus we conclude:
\begin{cor}\label{expcor}
Let $X\subseteq U$ be a fixed sets of $|X|=\balls$ balls and
let $h:U\to [\bins]$ be a simple tabulation hash function. Then
$\left|\E[\fbins]-\trm\right|/\trm=\Olog (\bins^{-1/c})$.
\end{cor}
As discussed above, the high probability bound from~\cite{Pat} takes over
when the bounds from Theorem~\ref{expthm} get weaker. This is because
the analysis in this paper is of a very different
nature than that in~\cite{Pat}.

\paragraph{Concentration of the number of non-empty bins:}
We now consider the concentration of $\fbins$ around its mean. 
%\mikcom{We do not have to assume that $\balls=\OO(\bins
%\log \bins)$, for in the proof, we can just complement with
%the results from~\cite{Pat} to see that the result also
%holds when $\balls=\omega(\bins \log \bins)$.}
In the fully random setting it was shown by Kamath et al.~\cite{Kam} that the concentration of $\fbins$ around $\mu_0$ is sharp: For any $\lambda\geq 0$ it holds that 
\begin{align*}
\Pr(|\fbins-\mu_0 |\geq \lambda)\leq 2\exp \left(-\frac{\lambda^2(\bins-1/2)}{\trm (2n-\trm)} \right)\leq 2 \exp \left(-\frac{\lambda^2}{2\trm} \right),
\end{align*}
which for example yields that $\fbins=\trm\pm \OO(\sqrt{\mu_0\log \bins})$ whp, that is, with probability $1-\OO(\bins^{-\gamma}$) for any choice of $\gamma=O(1)$. Unfortunately we cannot hope to obtain such a good concentration using simple tabulation hashing. To see this, consider the set of keys $[2]^\ell\times [\balls/2^\ell]$ for any constant $\ell$, e.g. $\ell=1$, and let $\mathcal{E}$ be the event that $h_i(0)=h_i(1)$ for $i=0,\dots,\ell-1$. This event occurs with probability $1/\bins^{\ell}$. Now if $\mathcal{E}$ occurs then the keys of $X_i=[2]^{\ell} \times \{i\}$ all hash to the same value namely $h_0(0)\oplus \cdots \oplus h_{\ell-1}(0) \oplus h_\ell(i)$. Furthermore, these values are independently and uniformly distributed in $[\bins]$ for $i\in [\balls/2^{\ell}]$ so the distribution of $\fbins$ becomes identical to the distribution of non-empty bins when $\balls/2^{\ell}$ balls are thrown into $\bins$ bins using truly random hashing. This observation ruins the hope of obtaining a sharp concentration around $\mu_0$ and shows that the lower bound in the theorem below is best possible being the expected number of non-empty bins when $\Omega(\balls)$ balls are distributed into $\bins$ bins. %For high probability upper bounds on the number of non-empty bins the situation is better since there is no natural event $\mathcal{E}'$ acting as a counterpart to $\mathcal{E}$ by causing the balls to ``spread out'' whp. 
\begin{thm}\label{hpthm}
Let $X\subseteq U$ be a fixed sets of $|X|=\balls$ keys.  Let $h:U\to [\bins]$ be a simple tabulation hash function. Then whp
\begin{align*}
\fbins \geq \bins \left(1-\left(1-\frac{1}{\bins}\right)^{\Omega(\balls)}\right) 
%\leq \min\left\{\balls,\bins,\mu_0+\OO \left(\sqrt{  \balls^{2-1/c}\log  \bins} \right)\right\}.
\end{align*}
%In particular $\fbins \leq \mu_0(1+\Olog(\balls^{-1/{(2c)}}))$ whp.
\end{thm}
%Note that the lower bound on $\fbins$ is the expected number of non-empty bins when $\Omega(\balls)$ balls are thrown uniformly and independently at random into $\bins$ bins. The explanation before the theorem shows that we may indeed end in this situation with probability $\Omega(\bins^{-\alpha})$ for some constant $\alpha$ and thus we cannot hope for at better high probability lower bound
%It is interesting to note that in the high probability upper bound, the relative error is $\Olog(\balls^{-1/{(2c)}})$ whereas the expected relative error from Corollary~\ref{expcor} was $\Olog(\bins^{-1/{c}})$.

As argued above, the lower bound in Theorem~\ref{hpthm} is optimal. Settling with a laxer requirement than high probability, it turns out however that $\fbins$ is somewhat concentrated around $\mu_0$. This is the content of the following theorem which also provides a high probability upper bound on $\fbins$.

\begin{thm}\label{concthm}
Let $X\subseteq U$ be a fixed sets of $|X|=\balls$ keys.  Let $h:U\to [\bins]$ be a random simple tabulation hash function. For  $t\geq 0$ it holds that
\begin{align}
&\Pr\left[\fbins\geq \mu_0+2t\right]=\OO\left( \exp\left(\frac{-t^2}{2m^{2-1/c}} \right) \right), &\text{and} \label{cbu}\\
&\Pr\left[\fbins\leq \mu_0-2t\right]=\OO\left( \exp\left(\frac{-t^2}{2\balls^{2-1/c}} \right)+\frac{\balls^2}{\bins t^2} \right). \label{cbl}&
\end{align}
\end{thm}
The term $\balls^2/(\bins t^2)$ in the second bound in the theorem may be unexpected but it has to be there (at least when $\balls=O(\bins)$) as we will argue after proving the theorem. %To see that there has to be ``something'' note that the set of keys $[2] \times [\balls/2]$ shows that some term must account for the fact that certain events (in this case $h_0(0)=h_0(1)$) that are relatively likely may cause a large number of keys to collide.

Theorem~\ref{concthm} is proved using Azuma's inequality (which we will state and describe later). It turns out that when $\balls \ll \bins$ one can obtain stronger concentration using a stronger martingale inequality. For intuition, the reader is encouraged to think of the fully random setting where $\balls$ balls are thrown sequentially into $\bins$ bins independently and uniformly at random: In this setting the allocation of a single ball can change the conditionally expected number of non-empty bins by at most $1$ and this is the type of observation that normally suggests applying Azuma's inequality. However, when $\balls \ll \bins$, it is unlikely that the allocation of a ball will change the conditional expectation of the number of non-empty bins by much --- for that to happen the ball has to hit a bin that is already non-empty, and the probability that this occurs is at most $\balls/\bins \ll 1$. Using a martingale inequality by Mcdiarmid~\cite{Colin}, that takes the variance of our martingale into consideration, one can obtain the following result which is an improvement over Theorem~\ref{concthm} when $\balls \ll \bins$, and matches within $O$-notation when $\balls=\Theta(\bins)$.
\begin{thm}\label{concthm2}
Let $X\subseteq U$ be a fixed sets of $|X|=\balls$ keys.  Let $h:U\to [\bins]$ be a random simple tabulation hash function. Assume $\balls\leq \bins$. For  $t\geq 0$ it holds that
\begin{align}
&\Pr\left[\fbins\geq \mu_0+t\right]= \exp\left(-\Omega\left( \min \left\{\frac{t^2}{\frac{\balls^{3-1/c}}{\bins}}, \frac{t}{\balls^{1-1/c}}  \right\}\right) \right), &\text{and} \label{cbu2}\\
&\Pr\left[\fbins\leq \mu_0-t\right]=\exp\left(-\Omega\left( \min \left\{\frac{t^2}{\frac{\balls^{3-1/c}}{\bins}}, \frac{t}{\balls^{1-1/c}}  \right\}\right) \right)+\OO \left(\frac{\balls^2}{\bins t^2} \right). \label{cbl2}&
\end{align}
\end{thm}
The above bounds are unwieldy  so let us disentangle them. 
First, one can show using simple calculus that when $2\leq \balls\leq \bins$ then $\trm=\balls-\Theta(\balls^2/\bins)$. If $\balls^{1+1/c}=o(\bins)$ we thus have that $\trm=\balls-o(\balls^{1-1/c})$. To get a non-trivial bound from~\eqref{cbu2} we have to let $t=\Omega(\balls^{1-1/c})$ and then  $\trm+t=m+\omega(\balls^{1-1/c})$. This means that~\eqref{cbu2} is trivial when $\balls^{1+1/c}=o(\bins)$ as we can never have more than $\balls$ non-empty bins. For comparison,~\eqref{cbu} already becomes trivial when $\balls^{1+1/(2c)}=o( \bins)$.

%As we have to let $t=\Omega(\balls^{1-1/c})$ to get a non-trivial bound, we thus have that~\eqref{cbu2} is trivial when $\balls^{1+1/c}=o( \bins)$:  In this case $\trm+t=\balls-\Theta(\balls^2/\bins)+ \Omega(\balls^{1-1/c})=\balls+\omega(1)$. 

 Suppose now that $\balls^{1+1/c}=\Omega( \bins)$. For a given $\delta$ put 
$$
t_0=\eta \max \left\{\sqrt{\frac{\balls^{3-1/c}}{\bins}\log \frac{1}{\delta}}, \balls^{1-1/c}\log \frac{1}{\delta} \right\},
$$
for some sufficiently large $\eta=O(1)$. Then~\eqref{cbu2} gives that $\Pr\left[\fbins\geq \mu_0+t_0\right]\leq \delta$. It remains to understand $t_0$:  Assuming that $\balls^{1+1/c}\geq \bins \log \frac{1}{\delta} $, we have that $t_0=\OO\left(\sqrt{\frac{\balls^{3-1/c}}{\bins}\log \frac{1}{\delta}}\right)$. For comparison, to get the same guarantee on the probability using~\eqref{cbu} we would have to put $t_0=\Omega\left(\sqrt{\balls^{2-1/c} \log \frac{1}{\delta}}\right)$, which is a factor of $\sqrt{\bins/\balls}$ larger.

Turning to~\eqref{cbl2}, it will typically in applications be the term $\OO \left(\frac{\balls^2}{\bins t^2} \right)$ that dominates the bound.  For a given $\delta$ we would choose $t=\max\{t_0,\balls/\sqrt{\bins \delta}\}$ to get $\Pr\left[\fbins\leq \mu_0-t\right]=\OO(\delta)$.

\subsection{Projecting into Arbitrary Ranges}\label{Projdessec}
Simple tabulation is an efficient hashing scheme for hashing into $r$-bit hash values. But what do we do if we want hash values in $[\bins]$ where $2^{r-1}<\bins<2^r$, say $\bins=3\times2^{r-2}$? Besides being of theoretical interest this is an important question in several practical applications. For example, when designing Bloom filters (which we will describe shortly), to minimize the false positive probability, we have to choose the size $\bins$ of the filters such that $\bins\approx \balls/\ln(2)$. When $\bins$ has to be a power of two, we may be up to a factor of $\sqrt{2}$ off, and this significantly affects the false positive probability. Another example is cuckoo hashing~\cite{Pagh}, which was shown in~\cite{Pat} to succeed with simple tabulation with probability $1-\OO(\bins^{-1/3})$ when $2\balls(1+\eps)\leq \bins$. If $\balls=2^r$ we have to choose $\bins$ as large as $2^{r+2}=4 \balls$ to apply this result, making it much less useful.

The way we remedy this is a standard trick, see e.g.~\cite{Thorup14}. We choose $r$ such that $2^r \gg \bins$, and hash in the first step to $r$-bit strings with a simple tabulation hash function $h:U\to [2^r]$. Usually $2^r\geq \bins^2$ suffices and then the entries of the character tables only becomes twice as long. Defining $s:[2^r]\to [\bins]$ by $s(y)=\lfloor y \bins/2^r \rfloor$ our combined hash function $U\to [\bins]$ is simply defined as $s \circ h$. Note that $s$ is very easy to compute since we do just one multiplication and since the division by $2^r$ is just an $r$-bit right shift. The only property we will use about $s$ is that it is \emph{most uniform} meaning that for $z\in [\bins]$ either, $|s^{-1}(\{z\})|=\lfloor \frac{2^r}{\bins} \rfloor$ or $|s^{-1}(\{z\})|=\lceil \frac{2^r}{\bins} \rceil$. For example, we could also use $s':[2^r]\to [\bins]$ defined by $s'(y)=y \pmod \bins$, but $s$ is much faster to compute. Note that if  $2^r\geq \bins^2$, then $\left| \frac{|s^{-1}(\{z\})|}{2^r}-\frac{1}{\bins} \right|\leq 2^{-r}\leq \bins^{-2}$.

\emph{A priori} it is not obvious that $s\circ h$ has the same good properties as ``normal'' simple tabulation. The set of bins can now be viewed as $\{s^{-1}(\{z\}):z \in [\bins]\}$, so each bin consists of many ``sub-bins'', and a result on the number of non-empty sub-bins does not translate directly to any useful result on the number of non-empty bins.
Nonetheless, many proofs of results for simple tabulation do not need  to be modified much in this new setting. For example, the simplified proof given by Aamand et al.~\cite{Aamand} of the result on cuckoo hashing from~\cite{Pat} can be checked to carry over to the case where the hash functions are implemented as described above if $r$ is sufficiently large. We provide no details here.

For the present paper the relevant analogue to Theorem~\ref{expthm} is the following:
\begin{thm}\label{expthm2}
Let $X\subseteq U$ be a fixed set of $|X|=\balls$ balls, and let $S\subseteq
[2^r]$ with $|S|/2^r=\rho$. Suppose $h:U\to [2^r]$ is a simple tabulation hash function. Define $\trp'=1-(1-\rho)^\balls$. If $p$ denotes the probability that
$h(X)\cap S\neq \emptyset$, then
\begin{align*}
|p-\trp'|\leq \balls^{2-1/c} \rho^2 %\quad \text{and hence} \quad \left|\E[\fbins]-\trm\right| \leq \frac{\balls^{2-1/c}}{\bins}.
\end{align*}
If we let $S$ (and hence $\rho$) depend on the hash of a distinguished query ball $q\in
U\backslash X$, then the bound on $p$ above is replaced by the
weaker $|p-\trp|\leq 2\balls^{2-1/c}\rho^2$.
\end{thm}
If we assume $2^r\geq \bins^2$, say, and let $S=s^{-1}(\{z\})$ be a bin of $S\subset[2^r]$ we obtain the following estimate on $p$:
\begin{align*}
|p-\trp|&\leq |p-\trp'|+|\trp'-\trp| \\
 &\leq \balls^{2-1/c}\left(\frac{1}{n}+\frac{1}{2^r}\right)^2+\frac{\balls}{2^r}=\frac{\balls^{2-1/c}}{\bins^2}(1+o(1))
\end{align*}

This is very close to what is obtained from Theorem~\ref{expthm} and to make the difference smaller we can increase $r$ further.

There are also analogues of Theorem~\ref{hpthm},~\ref{concthm} and~\ref{concthm2} in which the bins are partitioned into groups of almost equal size and where the interest is in the number of groups that are hit by a ball. To avoid making this paper unnecessarily technical, we refrain from stating and proving these theorems, but in Section~\ref{Projsec} we will show how to modify the proof of Theorem~\ref{expthm} to obtain Theorem~\ref{expthm2}.

\subsection{Alternatives} 
One natural alternative to simple tabulation is to use $k$-independent
hashing~\cite{WC}. Using an easy variation\footnotemark of an inclusion-exclusion based argument by
Mitzenmacher and Vadhan~\cite{Vadhan} one can show that if $k$ is odd
and if $\balls \leq \bins$ the probability $p$ that a given bin is
non-empty satisfies
\begin{align}\label{kindineq}
\trp -\OO \left(\left( \frac{\balls}{\bins} \right)^k \frac{1}{k!}
\right)\leq p \leq \trp +\OO \left(\left( \frac{\balls}{\bins}
\right)^{k+1} \frac{1}{(k+1)!} \right),
\end{align}
and this is optimal, at least when $k$ is not too large, say $k=o
(\sqrt{\balls})$ --- there exist two (different) $k$-independent
families making respectively the upper and the lower bound tight for a
certain set of $\balls$ keys. A similar result holds when $k$ is even.
Although $p$ approaches $\trp$ when $k$ increases, for $k=\OO(1)$
and $\balls =\Omega(\bins)$, we have a deviation by an additive constant
term. In contrast, the probability that a bin is non-empty when using
simple tabulation is asymptotically
the same as in the fully random setting.

Another alternative when studying the number of non-empty bins is to assume that the input comes with a certain amount of randomness.  This was studied in~\cite{Vadhan} too and a slight variation\footnotemark[\value{footnote}] of their argument shows that if the input $X\subseteq U$ has enough entropy the probability that a bin is empty is asymptotically the same as in the fully random setting even if we only use $2$-independent hashing. This is essentially what we get with simple tabulation. However, our results have the advantage of holding for any input with no assumptions on its entropy. %To our knowledge it is the only practical hash function known to achieve this.
 \footnotetext{Mitzenmacher and Vadhan
  actually estimate the probability of getting a false positive when
  using $k$-independent hashing for Bloom filters, but this error
  probability is strongly related to the expected number of
  non-empty bins $\E[\fbins]$ (in the fully random setting it
  \emph{is} $\E[\fbins]/\bins$). Thus only a slight modification of
  their proof is needed.}
Now~\eqref{kindineq} also suggests the third alternative of looking
for highly independent hash functions. For the
expectation~\eqref{kindineq} shows that if $\balls \leq \bins$ we
would need $k=\Omega(\log \bins /\log \log \bins)$ to get guarantees
comparable to those obtained for simple tabulation. Such highly
independent hash functions were first studied by Siegel~\cite{Sieg},
the most efficient known construction today being the double
tabulation by Thorup~\cite{Mik} which gives independence
$u^{\Omega(1/c^2)}\gg \log\bins$ using space $O(cu^{1/c})$ and time $O(c)$. 
While this space and time matches that of simple tabulation within constant factors, it
is slower by at least an order of magnitude. As mentioned
in~\cite{Mik}, double tabulation with 32-bit keys divided
into 16-bit characters requires 11 times as many character table
lookups as with simple tabulation and we lose the same factor in
space. The larger space of double tabulation means that tables
may expand into much slower memory, possibly costing us another order of magnitude
in speed. 

There are several other types of hash functions that one could
consider, e.g., those from~\cite{Dietz,2Pagh}, but simple tabulation
is unique in its speed (like two multiplications in the experiments
from~\cite{Pat}) and ease of implementation, making it a great choice
in practice. For a more thorough comparison of simple tabulation with other
hashing schemes, the reader is refered to~\cite{Pat}.

\section{Applications}
Before proving our main results we describe two almost immediate applications.
\subsection{Bloom Filters}\label{Bloomsec}
Bloom filters were introduced by Bloom~\cite{Bloom}. We will only discuss them briefly here and argue which guarantees are provided when implementing them using simple tabulation. For a thorough introduction including many applications see the survey by Broder and Mitzenmacher~\cite{Broder}. A Bloom filter is a simple data structure which space efficiently represents a set $X\subseteq U$ and supports membership queries of the form ``is $q$ in $X$''. It uses $k$ independent hash functions $h_0,\dots,h_{k-1}:U\to [\bins]$ and $k$ arrays $A_0,\dots,A_{k-1}$ each of $\bins$ bits which are initially all $0$. For each $x\in X$ we calculate $(h_i(x))_{i\in [k]}$ and set the $h_i(x)$'th bit of $A_i$ to $1$ noting that a bit may be set to $1$ several times. To answer the query ``is $q$ in $X$'' we check if the bits corresponding to $(h_i(q))_{i\in [k]}$ are all $1$, outputting ``yes'' if so and ``no'' otherwise. If $q\in X$ we will certainly output the correct answer but if $q\notin X$ we potentially get a false positive in the case that all the bits corresponding to $(h_i(q))_{i\in [k]}$ are set to $1$ by other keys in $X$. In the case that $q\notin X$ the probability of getting a false positive is  
\begin{align*}
\prod_{i=0}^{k-1} \Pr [h_i(q) \in h_i(X)],
\end{align*}
which with fully random hashing is $\trp^k= (1-(1-1/\bins)^\balls)^k \approx (1-e^{-\balls/\bins})^k$.% and it can be shown that to optimise this failure probability for fixed $\balls$ and $\bins$ we must choose $k$ such that $\trp \approx 1/2$.

It should be noted that Bloom filters are most commonly described in a related though not identical way. In this related setting we use a single $(k\bins)$-bit array $A$ and let $h_1,\dots,h_{k-1}:U\to [k\bins]$, setting the bits of $A$ corresponding to $(h_i(x))_{i\in [k]}$ to $1$ for each $x\in X$. With fully random hashing the probability that a bit is set to $1$ is then $q_0:=1-\left(1-\frac{1}{k\bins} \right)^{\balls k}$ and the probability of a false positive is thus at most $q_0^k=\left(1-\left(1-\frac{1}{k\bins} \right)^{\balls k}\right)^k\leq \trp^k$.  Despite the difference, simple calculus shows that $\trp-q_0=\OO(1/\bins)$ and so 
\begin{align*}
\trp^k-q_0^k=(\trp-q_0) \sum_{i=0}^{k-1}\trp^iq_0^{k-i-1}=\OO\left( \frac{k\trp^{k-1}}{\bins} \right).
\end{align*}
In particular if  $\trp =1-\Omega(1)$ or if the number of filters $k$ is not too large (both being the case in practice) the failure probability in the two models are almost identical. We use the model with $k$ different tables each of size $\bins$ as this makes it very easy to estimate the error probability using Theorem~\ref{expthm} and the independence of the hash functions. We can in fact view $h_i$ as a map from $U$ to $[k\bins ]$ but having image in $[(i+1)n]\backslash[i\bins]$ getting us to the model with just one array.

From Theorem~\ref{expthm} we immediately obtain the following corollary.
\begin{cor}
Let $X\subseteq U$ with $|X|=\balls$ and $y\in U\backslash X$. Suppose we represent $X$ with a Bloom filter using $k$ independent simple tabulation hash functions $h_0,\dots,h_{k-1}:U\to [\bins]$. The probability of getting a false positive when querying $q$ is at most 
\begin{align*}
 \left(\trp+\frac{2\balls^{2-1/c}}{\bins^2} \right)^k.  % \leq \trp ^k + \OO(kn^{-1/c}).
\end{align*}
\end{cor}
At this point one can play with the parameters. In the fully random
setting one can show that if the number of balls $\balls$ and the the
total number of bins $k\bins$ are fixed one needs to choose $k$ and
$\bins$ such that $\trp\approx 1/2$ in order to minimise the error
probability (see~\cite{Broder}). For this, one needs $\balls \approx
\bins \ln(2)$ and if $\bins$ is chosen so, the probability above is at
most $(\trp+\OO(\bins^{-1/c} ) )^k$. In applications, 
$k$ is normally a small number like $10$ for a 0.1\% false positive probability. In particular, $k=\bins^{o(1)}$, and then $(\trp+\OO(\bins^{-1/c} ) )^k=\trp^k(1+o(1))$, 
asymptotically matching the fully random setting.

To resolve the issue that the range of a simple tabulation function has size $2^r$ but that we wish to choose $\bins \approx
\balls /\ln(2)$, we choose $r$ such that $2^r\geq \bins^2$ and use the combined hash function $s\circ h:U\to [\bins]$ described in Section~\ref{Projdessec}. Now appealing to Theorem~\ref{expthm2} instead of Theorem~\ref{expthm} we can again drive the false positive probability down to $\trp^k(1+o(1))$ when $k=\bins^{o(1)}$.

\paragraph{Alternatives:} 
The argument by Mitzenmacher and Vadhan~\cite{Vadhan} discussed in
relation to~\eqref{kindineq} actually yields a tight bound on the
probability of a false positive when using $\ell$-independent hashing for
Bloom filters. We do not state their result here but mention that when
$\ell$ is constant the error probability may again deviate by an additive
constant from that of the fully random setting.  It is also shown
in~\cite{Vadhan} that if the input has enough entropy we can get the
probability of a false positive to match that from the fully random
setting asymptotically even using $2$-independent hashing, yet
it cannot be trusted for certain types of input. 

Now, imagine you are a software engineer that wants to implement a
Bloom filter, proportioning it for a desired low false-positive
probability. You can go to a wikipedia page
(\url{en.wikipedia.org/wiki/Bloom_filter}) or a texbook like
\cite{Mit} and read how to do it assuming full randomness. If you read~\cite{Vadhan}, what do you do? Do you set $\ell=2$ and cross your fingers, or do you pay the cost of
a slower hash function with a larger $\ell$, adjusting the false-positive 
probabilities accordingly? Which $\ell$ do you pick?

With our result, there are now hard choices. The answer is simple. We
just have to add that everything works as stated for any possible
input if the hashing is implemented with simple tabulation hashing
(\url{en.wikipedia.org/wiki/Tabulation_hashing}) which is both very
fast and very easy to implement.

\subsection{Filter Hashing}\label{filtsec}
In Filter hashing, as introduced by Fotakis et al.~\cite{Fot}, we wish to store as many elements as possible of a set $X\subseteq U$ of size $|X|=\balls=\bins$ in $d$ hash tables $(T_i)_{i\in [d]}$. The total number of entries in the tables is at most $\bins$ and each entry  can store just a single key. For $i\in [d]$ we pick independent hash functions $h_i:U\to [\bins_i]$ where $\bins_i$ is the number of entries in $T_i$. The keys are allocated as follows: We first greedily store a key from $h_0^{-1}(\{y\})$ in $T_0[y]$ for each $y\in h_0(X)$. This lets us store exactly $|h_0(X)|$ keys. Letting $S_0$ be the so stored keys and $X_1=X\backslash S_0$ the remaining keys, we repeat the process, storing $|h(X_1)|$ keys in $T_1$ using $h_1$ etc.

An alternative and in practice more relevant way to see this is to imagine that the keys arrive sequentially. When a new key $x$ arrives we let $i$ be the smallest index such that $T_i[h_i(x)]$ is unmatched and store $x$ in that entry. If no such $i$ exists the key is not stored. The name Filter hashing comes from this view which prompts the picture of particles (the keys) passing through filters (the tables) being caught by a filter only if there is a vacant spot.

The question is for a given $\eps>0$ how few filters that are needed
in order to store all but at most $\eps \bins$ keys with high
probability. Note that the remaining $\eps \bins$ keys can be stored
using any hashing scheme which uses linear space, for example
%\mikcom{Skulle vi ikke droppe: the two layered hashing scheme by
%  Fredman et al.~\cite{Fredman} or }
Cuckoo hashing with simple tabulation~\cite{Pagh,Pat}, to get a
  total space usage of $(1+O(\eps))\bins$.
  
One can argue that with fully random hashing one needs
$\Omega(\log^2(1/\eps))$ filters to achieve that whp
at least $(1-\eps)\bins$ keys are stored. To see that we can achieve
this bound with simple tabulation we essentially proceed as
in~\cite{Fot}. Let $\gamma>0$ be any constant and choose $\delta>0$ according to Theorem~\ref{hpthm}
so that if $X\subseteq
U$ with $|X|=\balls$ and $h:U\to [\bins]$ is a simple tabulation hash
function, then $\fbins \geq \bins(1-\left(1-1/\bins
\right)^{\delta\balls})$ with probability at least
$1-\bins^{-\gamma}$. 

Let $m_0=n$. For $i=0,1,\ldots$, we pick $n_i$ to be the largest power
of two below $\delta m_i/\log (1/\eps)$.  We then set $m_{i+1}=n-\sum_{j=0}^i n_j$,
terminating when $m_{i+1}\leq \eps n$.
Then $T_i$ is indexed by
$(\log_2 n_i)$-bit strings --- the range of a simple tabulation
hash function $h_i$.  Letting $d$ be minimal such that $\balls_d\leq \eps \bins$ we have that $(1-\eps)\bins\leq \sum_{i\in [d]} \bins_i\leq\bins$ and as  $\balls_i$ decreases by at least a factor of $\left(1-\frac{\delta}{2\log(1/\eps)}\right)$ in each step, $d\leq \lceil 2\log(1/\eps)^2/\delta\rceil$.
 
How many bins of $T_i $ get filled? Even if all bins from filters $(T_j)_{j<i}$ are non-empty we have at least $\balls_i$ balls left and so with probability $1-\OO(\bins_i^{-\gamma})$ the number of bins we hit is at least 
\begin{align*}
\bins_i(1-\left(1-1/\bins_i \right)^{\delta\balls_i})\geq \bins_i (1-e^{-\delta\balls_i/\bins_i})\geq \bins_i(1-\eps).
\end{align*}
Thus, with probability at least $1-\OO(d \bins_d^{-\gamma})$, for each $i\in [d]$, filter $i$ gets at least $(1-\eps)\bins_i$ balls. Since $ \sum_{i\in [d]}\bins_i\geq (1-\eps)\bins$, the number of overflowing balls is at most $2 \eps \bins$ in this case. Assuming for example that $\eps=\Omega(\bins^{-1/2})$, as would be the case in most applications, we get that the fraction of balls not stored is $\OO(\eps)$ with probability at least $1-\Olog(\bins^{-\gamma/2})$.

\begin{comment}

\paragraph{Det du havde skrevet}
Let $d=\lceil \log(1/\eps)^2/\delta\rceil$ and
for $i\in [d]$ choose the size of $T_i$ to be
\begin{align*}
\bins_i=\left\lfloor \frac{\delta\bins}{\log(1/\eps)}\left(1-\frac{\delta}{\log(1/\eps)} \right)^{i-1} \right\rfloor
\end{align*}
Then $(1-\eps)\bins-d\leq \sum_{i=1}^d \bins_i\leq\bins$. Moreovclaim that table $i$ gets at least $(1-\eps)\bins_i$ balls with probability at least $1-\bins_i^{-\gamma}$. To see this note that the worst case at the point where we are store keys in $T_i$ is if the tables $(T_j)_{j<i}$ are completely filled. In this worst case there are at least $\balls_i:=\bins \left(1-\frac{\delta}{\log(1/\eps)} \right)^{i-1}$ balls left when we are to fill $T_i$ and since $|T_i|=\bins_i\leq \delta\balls_i/\log(1/\eps)$ it follows from the choice of $\delta$ that with probability at least $1-\bins_i^{-\gamma}$ we hit at least 
\begin{align*}
\bins_i(1-\left(1-1/\bins_i \right)^{\delta\balls_i})\geq \bins_i (1-e^{-\delta\balls_i/\bins_i})\geq \bins_i(1-\eps)
\end{align*}
bins of $T_i$. This implies that with probability at least $1-\sum_{i\in [d]} \bins_i^{-\gamma}\geq 1-d \bins_d^{-\gamma}$ all the filters get at least $(1-\eps)\bins_i$ balls. If this occurs then the number of balls not allocated is no more than $2\eps\bins+d$. Assuming for example that $\eps=\Omega(\bins^{-1/2})$ as would be the case in most applications we get that the fraction of balls not stored is $\OO(\eps)$ with probability at least $1-\Olog(\bins^{-\gamma/2})$.

\end{comment}

\paragraph{Alternatives} The hashing scheme for Filter hashing described in~\cite{Fot} uses $(12\lceil \ln(4/\eps)+1 \rceil)$-independent polynomial hashing to achieve an overflow of at most $\eps \bins$ balls. In particular the choice of hash functions depends on $\eps$ and becomes more unrealistic the smaller $\eps$ is. In contrast when using simple tabulation (which is only $3$-independent) for Filter hashing we only need to change the number of filters, not the hashing, when $\eps$ varies. It should be mentioned that only $\lceil \ln(4/\eps)^2 \rceil$ filters are needed for the result in~\cite{Fot} whereas we need a constant factor more. It can however be shown  (we provide no details) that we can get down to $d=\lceil 2\log(1/\eps)^2 \rceil $ filters by applying~\eqref{cbl} of Theorem~\ref{concthm} if we settle for an error probability of $\OO(\bins^{-1+\eta})$ for a given constant $\eta > 0$.

Taking a step back we see the merits of a hashing scheme giving many
complementary probabilistic guarantees. As shown by P\v{a}tra\c{s}cu
and Thorup~\cite{Pat}, Cuckoo hashing~\cite{Pagh} implemented with
simple tabulation succeeds with probability $1-O(\bins^{-1/3})$ (for a
recent simpler proof of this result, see Aamand et
al. \cite{Aamand}). More precisely, for a set $X'$ of $\balls'$
balls, let $\bins'$ be the least power of two bigger than
$(1+\Omega(1))\balls'$. Allocating tables $T'_0,T'_1$ of size $\bins'$, and
using simple tabulation hash functions $h'_0,h'_1:U\to[\bins']$, with probability $1-O(\bins^{-1/3})$ Cuckoo hashing succeeds in placing the keys such that
every key $x\in X'$ is found in either $T'_0[h_0'(x)]$ or
$T'_1[h_1'(x)]$. In case it fails, we just try again with new random $h'_0,h'_1$.

We now use Cuckoo hashing to store the $\bins'=O(\eps \bins)$ keys
remaining after the filer hashing, appending the Cuckoo tables to the
filter tables so that $T_{d+i}=T'_i$ and $h_{d+i}=h'_i$ for $i=0,1$. Then $x\in X$
if and only if for some $i\in [d+2]$, we have $x=T_i[h_i(x)]$.  We
note that all these $d+2$ lookups could be done in parallel.
Moreover, as the
output bits of simple tabulation are mutually independent, the $d+2$
hash functions $h_i:U \to [2^{r_i}]$, $2^{r_i}=n_i$, can be
implemented as a single simple tabulation hash function $h:U\to
[2^{r_1+\dots+r_{d+2}}]$ and therefore all be calculated using just
$c=\OO(1)$ look-ups in simple tabulation character tables.

%Taking a step back we see the merits of a hashing scheme giving many complimentary probabilistic guarantees. As shown by P\v{a}tra\c{s}cu and Thorup~\cite{Pat}, Cuckoo hashing~\cite{Pagh}  implemented with simple tabulation succeeds with probability $1-O(\bins^{-1/3})$ when $\bins \geq 2(1+\eps)\balls$ for any constant $\eps>0$ (for a recent, simpler proof of this result, see Aamand et al. \cite{Aamand}).  In particular, if we store the $\OO(\eps)$ overflowing balls from Filter hashing using Cuckoo hashing with simple tabulation, we get a space usage of $(1+O(\eps))\balls$. Determining whether a query ball $q$ is in $X$ is done by calculating $d+2$ hash values and checking the corresponding table entries --- in other words we do not have to change the procedure in case of overflowing balls. As the output bits of simple tabulation are mutually independent the $d+2$ hash functions  $h_i:U \to [2^{r_i}]$ can be implemented as a single simple tabulation hash function $h:U\to [2^{r_1+\dots+r_{d+2}}]$ and therefore be calculated using just $c=\OO(1)$ look-ups.

%$\bins_i=\Omega(\bins^{1/2}/\log \bins)$ and hence that the error probability is $\Olog(\bins^{-\gamma/2})$.
\section{Preliminaries}
%We will now recall several results concerning simple tabulation hashing that we need for proving our main results. We start out by discussing the independence of simple tabulation.
As in~\cite{Mik} we define a \textbf{position character} to be an element $(j,a)\in [c]\times\Sigma$. Simple tabulation hash functions are initially defined only on keys in $U$ but we can extend the definition to sets of position characters $S=\{(i_{j},a_{j}):j\in [k]\}$ by letting $h(S)=\bigoplus_{j\in [k]}h_{i_j}(a_j)$. This coincides with $h(x)$ when the key $x\in U=[\Sigma]^c$ is viewed as the set of position characters $\{(i,x[i]):i\in [c]\}$.  

We start by describing an ordering of the position characters, introduced by P\v{a}tra\c{s}cu and Thorup~\cite{Pat} in order to prove that the number of balls hashing to a specific bin is Chernoff concentrated when using simple tabulation. If $X\subseteq U$ is a set of keys and $\prec$ is any ordering of the position characters $[c] \times \Sigma$ we for $\alpha\in[c] \times \Sigma$ define $X_\alpha=\{x\in X \ \vert \ \forall \beta \in [c]\times \Sigma: \beta \in x \Rightarrow \beta \preceq \alpha \}$. Here we view the keys as sets of position characters. Further define $G_{\alpha}=X_{\alpha} \backslash (\bigcup_{\beta \prec \alpha} X_{\beta})$ to be the set of keys in $X_\alpha$ containing $\alpha$ as a position character. P\v{a}tra\c{s}cu and Thorup argued that the ordering may be chosen such that the groups $G_{\alpha}$ are not too large.
\begin{lemma}[P\v{a}tra\c{s}cu and Thorup~\cite{Pat}]\label{orderlemma}
Let $X\subseteq U$ with $|X|=\balls$. There exists an ordering $\prec$ of the position characters such that $|G_{\alpha}|\leq \balls^{1-1/c}$ for all position characters $\alpha$. If $q$ is any (query) key in $X$ or outside $X$, we may choose the ordering such that the position characters of $q$ are first in the order and such that $|G_{\alpha}|\leq 2\balls^{1-1/c}$ for all position characters $\alpha$.
\end{lemma}
Let us throughout this section assume that $\prec$ is chosen as to satisfy the properties of Lemma~\ref{orderlemma}. A set $Y\subseteq U$ is said to be \textbf{$d$-bounded} if $|h^{-1}(\{z\})\cap Y|\leq d$ for all $z\in R$. In other words no bin gets more than $d$ balls from $Y$. 

 %that if the number of balls is not too large compared to the number of bins then for any  constant $\gamma$ all groups will be $d$-bounded for a constant $d$ depending only on $\gamma$. As this result will be of importance for our high probability bounds on $\fbins$ we restate it here.
\begin{lemma}[P\v{a}tra\c{s}cu and Thorup~\cite{Pat}]\label{dboundedlemma}%Er det helt korrekt?
Assume that the number of bins $\bins$ is at least $\balls^{1-1/(2c)}$. For any constant $\gamma$, and $d=\min\left\{2c(3+\gamma)^c, 2^{2c(3+\gamma)} \right\}$ all groups $G_\alpha$ are $d$-bounded with probability at least $1-\bins^{-\gamma}$.
\end{lemma}
Lemma~\ref{dboundedlemma} follows from another lemma from~\cite{Pat} which we restate here as we will use it in one of our proofs.
\begin{lemma}[P\v{a}tra\c{s}cu and Thorup~\cite{Pat}]\label{dboundedlemma2}
Let $\eps>0$ be a fixed constant and assume that  $\balls\leq \bins^{1-\eps}$. For any constant $\gamma$ no bin gets more than $\min \left(((1+\gamma)/\eps)^c,2^{(1+\gamma)/\eps}\right)=O(1)$ balls with probability  at least $1-n^{-\gamma}$.
\end{lemma}
Let us describe heuristically why we are interested in the order $\prec$ and its properties. We will think of $h$ as being uncovered stepwise by fixing $h(\alpha)$ only when $(h(\beta))_{\beta \prec \alpha}$ has been fixed. At the point where $h(\alpha)$ is to be fixed the internal clustering of the keys in $G_\alpha$ has been settled and $h(\alpha)$ acts merely as a translation, that is, as a shift by an XOR with $h(\alpha)$. This viewpoint opens up for sequential analyses where for example it  may be possible to calculate the probability of a bin becoming empty or to apply martingale concentration inequalities. The hurdle is that the internal clustering of the keys in the groups are not independent as the hash value of earlier position characters dictate how later groups cluster so we still have to come up with ways of dealing with these dependencies.

\section{Proofs of main results}
%Paving the way for the proofs of our main results we start by stating two technical lemmas. We provide the proofs of both of these lemmas by the end of this section.  First some some notation: Consider the order $\prec$ on the keys from Lemma~\ref{orderlemma} such that the position characters of $q$ are first in this ordering and such that $|G_\alpha|\leq 2m^{1-1/c}$ for all position characters $\alpha$. Let $\alpha_1,\dots,\alpha_k$ denote the position characters enumerated so that $\alpha_1\prec \cdots \prec \alpha_k$. We will from this point on write $G_i$ in place of $G_{\alpha_i}$. 

%Let $C_i$ denote the number of sets $\{x,y\}\subseteq G_i$ such that $x \neq y$ but $h(x)=h(y)$~i.e. the number of pairs of colliding keys internal to $G_i$. Let also $C=\sum_{i=1}^kC_i$ denote the total number of collisions internal in the groups. The first lemma bounds the expected value of $C$ as well as its variance.
%\begin{lemma}\label{varlem}
%As above let $C$ denote the number of internal collisions in the groups. %Suppose further that $m^{1-2/c}\leq n$. 
%Then 
%\begin{align}
%\E[C] &\leq \frac{\balls^{2-1/c}}{\bins} \label{expcol}\\
%\var[C]&\leq \frac{4\balls^{3-2/c}}{\bins^2}+\frac{(3^c+1)\balls^2}{\bins} \label{varcol}
%\end{align}
%If in particular $\balls^{1-2/c}=\OO( \bins)$ we have that $\var[C]=\OO(\balls^2/\bins)$
%\end{lemma}
In order to pave the way for the proofs of our main results we start by stating two technical lemmas, namely Lemma~\ref{ineqlemma} and~\ref{varlem} below. We provide proofs at the end of this section. Lemma~\ref{ineqlemma} is hardly more than an observation. We include it as we will be using it repeatedly in the proofs of our main theorems.

\begin{lemma}\label{ineqlemma} %More general result?
Assume  $\alpha\geq 1$ and $m,m_0\geq0$ are real numbers. Further assume that  $ 0\leq g_1,\dots,g_k \leq \balls_0$ and $\sum_{i=1}^kg_i=m$. Then \begin{align}\label{simineq1}
\sum_{i=1}^{k}g_i^\alpha\leq \balls_0^{\alpha-1}\balls.
\end{align}
If further $\balls_0\leq \bins$ for some real $\bins$ then 
\begin{align}\label{simineq2}
\prod_{i=1}^k \left(1- \frac{g_i}{\bins} \right)\geq \left(1-\frac{\balls_0}{\bins} \right)^{\balls/\balls_0}.
\end{align}
\end{lemma}
In our applications of Lemma~\ref{ineqlemma}, $g_1,\dots,g_k$ will be the sizes of the groups $G_{\alpha}$ described in Lemma~\ref{orderlemma}, and $\balls_0$ will be the upper bound on the group sizes provided by the same lemma. 

For the second lemma we assume that the set of keys $X$ has been partitioned into $k$ groups $(X_i)_{i\in [k]}$. Let $C_i$ denote the number of sets $\{x,y\}\subseteq X_i$ such that $x \neq y$ but $h(x)=h(y)$, that is, the number of pairs of colliding keys internal to $X_i$. Denote by $C=\sum_{i=1}^kC_i$ the total number of collisions internal in the groups. The second lemma bounds the expected value of $C$ as well as its variance in the case where the groups are not too large.
\begin{lemma}\label{varlem}
Let $X\subseteq U$ with $|X|=\balls$ be partitioned as above. Suppose that there is an $\balls_0\geq 1$ such that for all $i\in[k]$, $|X_i|\leq \balls_0$. %Suppose further that $m^{1-2/c}\leq n$. 
Then 
\begin{align}
\E[C] &\leq \frac{\balls \cdot \balls_0}{2 \bins}, & \text{and} \label{expcol}\\
\var[C]&\leq\frac{(3^c+1)\balls^2}{\bins} + \frac{\balls \cdot \balls_0^2}{\bins^2}.& \label{varcol}
\end{align}
For a given query ball $q\in U\backslash X$ and a bin $z\in [\bins]$, the upper bound on $\E[C]$ is also an upper bound on $\E[C \mid h(q)=z]$. For the variance estimate note that if in particular $\balls_0^2=\OO(\balls \bins)$, then $\var[C]=\OO(\balls^2/\bins)$. 
\end{lemma}
We will apply this lemma when the $X_i$ are the groups arising from the order $\prec$ of Lemma~\ref{orderlemma}.
With these results in hand we are ready to prove Theorem~\ref{expthm}.
\begin{proof}[\textbf{Proof of Theorem~\ref{expthm}}]
Let us first prove the theorem in the case where $y$ is a fixed bin not chosen dependently on the hash value of a query ball. If $\balls^{1-1/c}\geq \bins$ the result is trivial as then the stated upper bound is at least $1$. Assume then that $\balls^{1-1/c}\leq \bins$. Consider the ordering $\alpha_1\prec \cdots \prec \alpha_k$ of the position characters obtained from~Lemma~\ref{orderlemma} such that all groups $G_i:=G_{\alpha_i}$ have size at most $\balls^{1-1/c}$. We will denote by $\balls_0:=\balls^{1-1/c}$ the maximal possible group size.

We randomly fix the $h(\alpha_i)$ in the order obtained from $\prec$ not fixing $h(\alpha_i)$ before having fixed $h(\alpha_j)$ for all $j<i$. If $x\in G_i$ then $h(x)=h(\alpha_i) \oplus h(x \backslash \{\alpha_i\})$ and since $\beta \prec \alpha_i$ for all $\beta \in x \backslash \{\alpha_i\}$ only $h(\alpha_i)$ has to be fixed in order to settle $h(x)$. The number of different bins hit by the keys of $G_i$ when fixing $h(\alpha_i)$ is thus exactly the size of the set $\{h(x \backslash \{\alpha_i\}):x\in G_i\}$ which is simply translated by an XOR with $h(\alpha_i)$ and for  $x\in G_i$ we have that $h(x)$ is uniform in its range when conditioned on the values $(h(\alpha_j))_{j<i}$.

To make it easier to calculate the probability that $y\in h(X)$ we introduce some \emph{dummy balls}. At the point where we are to fix $h(\alpha_i)$ we dependently on $(h(\alpha_j))_{j< i}$
in any deterministic way choose a set $D_i\subseteq R=[\bins]$ of dummy balls, disjoint from $\{h(x\backslash \{\alpha_i\}): x \in G_i\}$, such that $\{h(x\backslash \{\alpha_i\}): x \in G_i\}\cup D_i$ has size exactly $|G_i|$. We will say that a bin $z$ is \emph{hit} if either $z\in h(X)$ or there exists an $i$ such that $z=d\oplus h(\alpha_i)$ for some $d \in D_i$. In the latter case we will say that $z$ is hit by a dummy ball. 
This modified random process can be seen as ensuring that when we are to finally fix the hash values of the elements of $G_i$ by the last translation with $h(\alpha_i)$, we modify the group by adding dummy balls to ensure that exactly $ |G_i|$ bins are hit by either a ball in $G_i$ or a dummy ball in $D_i$. We let $D=\sum_{i=1}^k|D_i|$ denote the total number of dummy balls.

Let $\mathcal{H}$ denote the event that $y$ is hit and $\mathcal{D}$ denote the event that $y$ is hit by a dummy ball. With the presence of the dummy balls, $\Pr[\mathcal{H}]$ is easy to calculate:
\begin{align*}
\Pr[\mathcal{H}]=1-\prod_{i=1}^k\left(1- \frac{|G_i|}{\bins} \right)\geq 1-\prod_{i=1}^k\left(1- \frac{1}{\bins} \right)^{|G_i|}=\trp.
\end{align*}
Clearly $\Pr[y\in h(X)]\geq \Pr[\mathcal{H}]-\Pr[\mathcal{D}]$ so for a lower bound on $\Pr[y\in h(X)]$ it suffices to  upper bound $\Pr[\mathcal{D}]$. Let $\mathcal{D}_i$ denote the event that $y$ is hit by a dummy ball from $D_i$. We can calculate $\Pr[\mathcal{D}_i]=\sum_{\ell=0}^\infty \Pr[\mathcal{D}_i \mid |D_i|= \ell]\times \Pr[|D_i|=\ell]$. The conditional probability $\Pr[\mathcal{D}_i \mid |D_i|= \ell]$ is exactly $\ell/\bins$ as the choice of $D_i$ only depends on the hash values $(h(\alpha_j))_{j<i}$ and when translated by an XOR with $h(\alpha_i)$ the bin $y$ is hit with probability $|D_i|/\bins$. It follows that $\Pr[\mathcal{D}_i]=\E[|D_i|]/\bins$ and thus that $\Pr[\mathcal{D}]\leq\sum_{i=1}^k\Pr[\mathcal{D}_i]=\E[D]/\bins$. Finally the total number of dummy balls is upper bounded by the number $C$ of internal collisions in the groups, so Lemma~\ref{varlem} gives that $\Pr[\mathcal{D}] \leq \E[C]/\bins \leq \frac{\balls^{2-1/c}}{2\bins^2}$.
%\begin{align*}
%p=\Pr[y\in h(X)]\geq\Pr[\mathcal{H}]-\Pr[\mathcal{D}]\geq \trp-\frac{\E[C]}{\bins}\geq \trp- \frac{\balls \cdot \balls_0}{2\bins} =\trp- \frac{\balls^{2-1/c}}{2\bins^2}
%\end{align*}
This gives the desired lower bound on $p$ (throwing away the factor of $1/2$, in order to simplify the statement in the theorem).  

For the upper bound note that $\Pr[y\in h(X)]\leq \Pr[\mathcal{H}]$ so by Lemma~\ref{ineqlemma}
\begin{align*}
p\leq \Pr[\mathcal{H}] \leq 1- \left(1- \frac{\balls_0}{\bins} \right)^{\balls/\balls_0}.
\end{align*}
%first note that $\Pr[y\in h(X)]\leq \Pr[\mathcal{H}]$. Second one can show using arguments similar to those used for proving Lemma~\ref{ineqlemma} that $\prod_{i=1}^k\left(1- \frac{|G_i|}{\bins} \right)\geq \left(1- \frac{\balls^{1-1/c}}{\bins} \right)^{\balls^{1/c}}$ and so %Provide arguments
%\begin{align*}
%p\leq \Pr[\mathcal{H}] \leq 1- \left(1- \frac{\balls^{1-1/c}}{\bins} \right)^{\balls^{1/c}}
%\end{align*}
Using the inequality $\left(1+\frac{x}{\ell} \right)^\ell\geq e^x\left(1-\frac{x^2}{\ell} \right)$ holding for $\ell\geq 1$ and $|x|\leq \ell$ with $x=-\balls/\bins$ and $\ell=\balls/\balls_0$ (note that $|x|\leq \ell$ as we assumed that $\balls^{1-1/c}\leq \bins$) we obtain that
\begin{align*}
p\leq 1-e^{-\balls/\bins} \left(1- \frac{\balls \cdot \balls_0}{\bins^2} \right)\leq 1-e^{-\balls/\bins} + \frac{\balls \cdot \balls_0}{\bins^2}\leq \trp+\frac{\balls^{2-1/c}}{\bins^2},
\end{align*}
as desired. The bound on $\E[\fbins]$ follows immediately as $\E[\fbins]=\sum_{y\in [\bins]} \Pr[y\in h(X)]$.

Finally consider the case where $y$ is chosen conditioned on $h(q)=z$ for a query ball $q\notin X$ and a bin $z$. Here we may assume that $2\balls^{1-1/c}\leq \bins$ as otherwise the claimed upper bound is at least $1$. We choose the ordering $\prec$ such that the position characters of $q$ are first in the order and such that all groups have size at most $2\balls^{1-1/c}$ which is possible by Lemma~\ref{orderlemma}. Let $\balls_0=\min(\balls,2\balls^{1-1/c})$ denote the maximal possible group size.  Introducing dummy balls the same way as before and repeating the arguments, the probability of the event $\mathcal{H}$ that $y$ is hit satisfies
\begin{align*}
\trp \leq \Pr[\mathcal{H} \mid h(q)=z]\leq 1-\left(1- \frac{\balls_0}{\bins} \right)^{\balls/\balls_0}\leq 1-e^{-\balls/\bins} \left(1- \frac{\balls \cdot \balls_0}{\bins^2} \right)\leq \trp+\frac{2\balls^{2-1/c}}{\bins^2}.
\end{align*}
The desired upper bound follows immediately as $\Pr[y\in h(X) \mid h(q)=z]\leq \Pr[\mathcal{H} \mid h(q)=z]$. \\
For the lower bound we again let $\mathcal{D}$ denote the event that $y$ is hit by a dummy ball and $\mathcal{D}_i$ denote the event that $y$ is hit by a dummy ball from $D_i$. Then
 %we get that the probability that $y$ is hit (by an $x\in X$ or a dummy ball) is $1-\prod_{i=1}^k\left(1- \frac{|G_i|}{\bins} \right)\geq \trp$ so for the lower bound it suffices to show that the probability that $y$ is hit by a dummy ball is at most $\frac{\balls^{2-1/c}}{\bins}$. Letting $\mathcal{D}$ denote the event that $y$ is hit by a dummy ball and $\mathcal{D}_i$ denote the event that $y$ is hit by a dummy ball from $D_i$ we have that 
\begin{align*}
\Pr[\mathcal{D}_i\mid h(q)=z] = \sum_{\ell=0}^\infty \Pr[\mathcal{D}_i \mid h(q)=z \wedge|D_i|=\ell] \times\Pr[|D_i|=\ell \mid h(q)=z].
\end{align*}
As before we have that $\Pr[\mathcal{D}_i \mid h(q)=z \wedge|D_i|=\ell] =\ell/\bins$ since the hash values of the position characters of $q$ are fixed before $h(\alpha_i)$. Thus,
\begin{align*}
\Pr[\mathcal{D}_i\mid h(q)=z]=  \sum_{\ell=0}^\infty \frac{\ell}{\bins}\Pr[|D_i|=\ell \mid h(q)=z]=\frac{\E[|D_i|\mid h(q)=z]}{\bins},
\end{align*}
and another union bound gives that 
\begin{align*}
\Pr[\mathcal{D}\mid h(q)=z]\leq \frac{\E[D \mid h(q)=z]}{\bins}\leq \frac{\E[C \mid h(q)=z]}{\bins}\leq \frac{\balls^{2-1/c}}{\bins^2},
\end{align*}
where we in the last step used Lemma~\ref{varlem}.
\end{proof}
We are now going to prove Theorem~\ref{hpthm}. We start out by recalling Azuma's inequality.
\begin{thm}[Azuma's inequality~\cite{Azuma}]
Suppose $(X_i)_{i=0}^k$ is a martingale satisfying that $|X_i-X_{i-1}|\leq s_i$ almost surely for all $i=1,\dots,k$. Let $s=\sum_{i=1}^k s_i^2$. Then for any $t\geq 0$ it holds that
\begin{align*}
\Pr(X_k\geq X_0+t) \leq \exp \left( \frac{-t^2}{2s}\right), \quad \text{and} \quad
\Pr(X_k\leq X_0-t) \leq \exp \left( \frac{-t^2}{2s}\right).
\end{align*}
 \end{thm}
 To apply Azuma's inequality we need to recall a little measure theory. Suppose $(\Omega,\mathcal{F},\Pr)$ is a finite measure space (that is $\Omega$ is finite), and that $Y:\Omega\to \R$ is an $\mathcal{F}$-measurable random variable. A sequence of $\sigma$-algebras $(\mathcal{F}_i)_{i=1}^k$ on $\Omega$ is called a \emph{filter} of the $\sigma$-algebra $\mathcal{F}$ if $\{\emptyset,\Omega\}=\mathcal{F}_0\subseteq \cdots \subseteq \mathcal{F}_k=\mathcal{F}$. Defining $Y_i=\E[Y|\mathcal{F}_i]$, the sequence $(Y_i)_{i=0}^k$ becomes a martingale with $Y_0=\E[Y]$ and $Y_k=Y$. It is for such martingales that we will apply Azuma's inequality.
 
 \begin{proof}[\textbf{Proof of Theorem~\ref{hpthm}}]
By the result by P\v{a}tra\c{s}cu and Thorup~\cite{Pat} we may assume that $\balls\leq C\bins \log \bins$ for some constant $C$  as otherwise all bins are full whp from which the results of the theorem immediately follow.

%The high probability upper bound on $\fbins$ follows directly from Theorem~\ref{concthm} which we are to prove shortly. Indeed for any constant $\gamma>0$ we can put $t=\sqrt{\gamma \cdot 2\balls^{2-1/c} \log \bins}$ in~\eqref{cbu} and obtain that $\fbins\leq \trm+2t =\trm +\OO \left(\sqrt{  \balls^{2-1/c}\log  \bins} \right)$ with probability $1-\OO(\bins^{-\gamma})$.

Let $G_1,\dots,G_k$ be the groups described in Lemma~\ref{orderlemma} and $\alpha_1,\dots,\alpha_k$ be the corresponding position characters. Again we think of the $h(\alpha_i)$ as being fixed sequentially. 
We let $(\Omega,\mathcal{F},\Pr)$ be the underlying probability space when choosing $h$, that is, $\Omega$ is the set of all simple tabulation hash functions, $\mathcal{F}=\mathcal{P}(\Omega)$, and $\Pr$ is the uniform probability measure on $\Omega$. For $i=0,\dots,k$ we define $\mathcal{F}_i=\sigma(h(\alpha_1),\dots, h(\alpha_i))$ to be the $\sigma$-algebra generated by the hash values of the first $i$ position characters. Then $\{\emptyset,\Omega\}=\mathcal{F}_0\subseteq \cdots \subseteq \mathcal{F}_k=\mathcal{F}$ is a filter of $\mathcal{F}$.

Ideally we would hope that for the martingale $(X_i)_{i=0}^k=(\E[\fbins \, | \, \mathcal{F}_i])_{i=0}^k$ we could effectively bound $|X_i-X_{i-1}|$ and thus apply Azuma's inequality. This is however too much to hope for --- the example with keys $[2]\times [\balls/2]$ shows that the hash value of a single position character can have a drastic effect on the conditionally expected number of non-empty bins. To remedy this we will again be using dummy balls but this time in a different way.

First of all, we let $\gamma>0$ be any constant. Since $\balls\leq C \bins \log \bins$, Lemma~\ref{dboundedlemma} gives that there exists a $d=d(\gamma)=O(1)$ such that all groups are $d$-bounded with probability at least $1-\bins^{-\gamma}$. Here is how we use the dummy balls: After having fixed $(h(\alpha_j))_{j\prec i}$ we again look at the set $G_i':=\{h(x\backslash \{\alpha_i\}): x \in G_i\}$ letting $I^{-}=\{i\in [k]: |G_i'|\geq \lceil |G_i|/d \rceil\}$ and $I^{+}=\{i\in [k]: |G_i'|< \lceil |G_i|/d \rceil\}$. For $i\in I^{-}$ we dependently on $(h(\alpha_j))_{j\prec i}$ choose a set $D_i^{-}\subseteq G_i'$ such that $|G_i' \backslash D_i^{-}|=\lceil |G_i|/d \rceil$. Similarly we for $i\in I^{+}$ choose a set $D_i^{+}\subseteq R$ disjoint from $G_i'$ such that $|G_i'\cup D_i^{+}|=\lceil |G_i|/d \rceil$. We say that bin $z$ is \emph{hit} if there exists an $i$ such that either
\begin{itemize}
\item[\textbf{1.}]  $i\in I^{-}$ and $z=y\oplus h(\alpha_i)$ for some $y\in G_i'\backslash D_i^{-}$, or
\item[\textbf{2.}] $i\in I^{+}$ and $z=y\oplus h(\alpha_i)$ for some $y\in G_i'\cup D_i^{+}$.
\end{itemize} 
This modified random process obtained by adding balls if $|G_i'|$ is too large and removing balls if it is too small can be seen as ensuring that when we are to finally fix the hash values of the elements of $G_i$ by the last translation by $h(\alpha_i)$ we first modify the group to ensure that we hit \emph{exactly} $\lceil |G_i|/d \rceil$ bins.

Importantly, we observe that if $G_i$ is $d$-bounded then $|G_i'|\geq  |G_i|/d $ and since $|G_i'|$ is integral $|G_i'|\geq  \lceil |G_i|/d \rceil$. Thus if all groups are $d$-bounded $I^{+}=\emptyset$, and no dummy balls are added.

Letting $H$ denote the number of bins hit, we have that 
\begin{align*}
\E[H]&=\bins \left(1-\prod_{i=1}^k \left(1-\frac{\lceil |G_i|/d \rceil}{\bins} \right) \right)\geq \bins \left(1-\prod_{i=1}^k \left(1-\frac{1}{\bins} \right)^{\lceil |G_i|/d \rceil} \right)\geq \bins \left(1-\left(1-\frac{1}{\bins} \right)^{\balls/d} \right).
\end{align*}
We now wish to apply Azuma's inequality to the martingale $(H_i)_{i=0}^k=(\E[H \mid \mathcal{F}_i])_{i=0}^k$. To do this we require a good upper bound on $|H_i-H_{i-1}|$ and we claim that in fact $|H_i-H_{i-1}|\leq |G_i|$. 
To see this, let the random variable $N_i$ denote the number of bins not hit when the hash values of the first $i$ position characters has been settled. Then  $H_i=\bins-N_i\prod_{j>i} \left(1-\frac{\lceil |G_j|/d \rceil}{\bins} \right)$ and so 
\begin{align*}
|H_i-H_{i-1}|=\prod_{j>i} \left(1-\frac{\lceil |G_j|/d \rceil}{\bins} \right)\left|N_i- \left(1-\frac{\lceil |G_i|/d \rceil}{\bins}\right)N_{i-1} \right|\leq   \left|N_i-N_{i-1}+\frac{\lceil |G_i|/d \rceil\cdot N_{i-1}}{\bins} \right|.
\end{align*}
% Since all $|G_i'|=|G_i|$ independently of the hash values of the position characters we have that 
%\begin{align*}
%|H_i-H_{i-1}|\leq |G_i|
%\end{align*}
Now $N_{i-1}-\lceil |G_i|/d \rceil\leq N_i\leq N_{i-1}$ as at least $0$ and most $\lceil |G_i|/d \rceil$ bins are hit after fixing $h(\alpha_i)$ and from this it follows that $|H_i-H_{i-1}| \leq\lceil |G_i|/d \rceil\leq |G_i|$.

Letting $s_i=|G_i|$ we have that $\sum_{i=1}^ks_i^2\leq \balls^{2-1/c}$ by Lemma~\ref{ineqlemma} and thus we can apply Azuma's inequality to obtain that
\begin{align*}
\Pr(H\leq \E[H]-t)\leq \exp\left( \frac{-t^2}{2\balls^{2-1/c}} \right).
\end{align*}
Putting $t=\sqrt{\gamma \cdot2\balls^{2-1/c} \log \bins}$ we obtain that with probability at least $1-\bins^{-\gamma}$
\begin{align*}
H\geq  \ell(\balls,\bins) :=\bins \left(1-\left(1-\frac{1}{\bins} \right)^{\balls/d} \right)-\sqrt{\gamma \cdot2\balls^{2-1/c} \log \bins}.
\end{align*}
As $I^{+}= \emptyset$ with probability at least $1-\bins^{-\gamma}$ and as we in this case have that $\fbins\geq H$ we have that $\fbins \geq \ell(\balls,\bins)$ with probability at least $1-2n^{-\gamma}$. 

The remaining part of proof is just combining what we have together with a little calculus! We first consider the case $\balls \leq \bins$. In this case the lower bound simply states that $\fbins =\Omega(m)$. To see that this bound holds observe that if (for example) $\balls \leq \bins^{1/2}$ then by Lemma~\ref{dboundedlemma2} no bin gets more than a constant number of balls with probability at least $1-\bins^{-\gamma}$. In particular $\fbins =\Omega(m)$ with probability at least $1-\bins^{-\gamma}$. If on the other hand $\balls \geq \bins^{1/2}$ then $\sqrt{\gamma \cdot2\balls^{2-1/c} \log \bins}=o(m)$ and $\ell(\balls,\bins)=\Omega(\balls)-o(\balls)=\Omega(\balls)$ which again gives the desired result.

Finally suppose $\balls \geq \bins$. Let $\alpha:=\left(1-1/\bins\right)^{\bins/(2d)}\leq e^{-1/(2d)}$ and let $\beta$ be a constant so large that $\beta\geq 2d$ and $\bins\left( 1-1/\bins\right)^{\balls/\beta}\geq \frac{1}{1-\alpha} \sqrt{\gamma \cdot2\balls^{2-1/c} \log \bins}$, the last requirement being possible as we assumed $\balls\leq C\bins \log \bins$. Then
\begin{align*}
\frac{\ell(\balls,\bins)}{\bins}&\geq 1-\left(1-\frac{1}{\bins} \right)^{\balls/d} -(1-\alpha)\left( 1-\frac{1}{\bins}\right)^{\balls/\beta} \\
&\geq 1-\left(1-\frac{1}{\bins} \right)^{\balls/\beta}\left(\left(1-\frac{1}{\bins} \right)^{\balls /(2d) } +(1-\alpha)\right) \geq 1-\left(1-\frac{1}{\bins} \right)^{\balls/\beta}. 
\end{align*}
Since $\fbins \geq \ell(\balls,\bins)$ with probability at least $1-2n^{-\gamma}$ this gives the desired result.
\end{proof}

We now prove Theorem~\ref{concthm}.
\begin{proof}[\textbf{Proof of Theorem~\ref{concthm}}]
When $\balls^{1-1/(2c)}\geq \bins$ the probability bounds of the theorem are trivial since they are $\Omega(1)$ when $t\leq n$ . We therefore assume henceforth that $\balls^{1-1/(2c)}\leq \bins$.
%If $\bins \leq \balls^{1-1/c}$ it follows from from the Chernoff concentration proved by Thorup and P\v{a}tra\c{s}cu~\cite{Pat} that all bins are non-empty whp and from this the two bounds in the theorem immediately follows. We will thus assume that $\balls = \OO(\bins \log \bins)$.  

Again consider the order $\prec$ obtained from Lemma~\ref{orderlemma} such that for all $i$ we have $|G_i|\leq\balls^{1-1/c}$. We again think of the hash values of the position characters as being fixed in the order obtained from $\prec$. We also introduce dummy balls in exactly the same way as we did in the proof of Theorem~\ref{expthm} using the same definition of a bin being hit. 

%More precisely when we are to fix $h(\alpha_i)$ the clustering and in particular the size of $h(G_i)$ has been settled and $h(\alpha_i)$ acts merely a translation of the values in $\{h(x\backslash \{\alpha_i\}): x \in G_i\}$. Dependently on $(h(\alpha)_j)_{j\prec i}$ we may thus (in any deterministic way) choose a set $D_i\subseteq R=[\bins]$ of dummy balls such that $G_i'=\{h(x\backslash \{\alpha_i\}): x \in G_i\}\cup D_i$ has size exactly $|G_i|$ which is also the size of $h(G_i)\cup \{y\oplus h(\alpha_i):y\in D_i\}$. We will say that a bin $z$ is \emph{hit} if either $z\in h(X)$ or there exists an $i$ such that $z=y\oplus h(\alpha_i)$ for some $y \in D_i$. In the latter case we will say that $z$ is hit by a dummy ball.
Letting $H$ denote the number of bins hit (by an $x\in X$ or a dummy ball) we have that $\E[H]=\bins\left( 1-\prod_{i=1}^k\left(1- \frac{|G_i|}{\bins} \right)\right)$, like in the proof of Theorem~\ref{expthm}, and
\begin{align*}
\trm \leq \E[H] \leq \trm+ \frac{\balls^{2-1/c}}{\bins}. %e^{-\balls/\bins}.
\end{align*}
Furthermore letting $\mathcal{F}_i=\sigma(h(\alpha_1),\dots,h(\alpha_i))$ be the $\sigma$-algebra generated by $(h(\alpha_j))_{j\leq i}$, the same argument as in the proof of Theorem~\ref{hpthm} gives that $H_i=\E[H | \mathcal{F}_i]$ is a martingale satisfying that $|H_{i}-H_{i-1}|\leq |G_i|$ for all $i$.
We can thus apply Azuma's inequality with $s_i=|G_i|$ and $s=\sum_{i=1}^k s_i^2\leq \balls^{2-1/c}$ (here we used Lemma~\ref{ineqlemma}) to obtain that 
\begin{align}
&\Pr[H\geq \E[H]+ t] \leq \exp \left(\frac{-t^2}{2\balls^{2-1/c}} \right), \quad \text{and} \label{azb1}\\
&\Pr[H\leq \E[H]- t] \leq \exp \left(\frac{-t^2}{2\balls^{2-1/c}} \right) \label{azb2}.
\end{align}
We now wish to translate this concentration result on the number of bins hit when the dummy balls are included to a concentration result on $\fbins$. We begin with the bound in~\eqref{cbu}. As $\fbins\leq H$ it suffices to bound the probability $\Pr[H\geq \trm +2t]$. Since $\E[H] \leq \trm+ \frac{\balls^{2-1/c}}{\bins}$,
\begin{align*}
\Pr[H\geq \trm +2t]\leq \Pr\left[H-\E[H]\geq 2t-\frac{\balls^{2-1/c}}{\bins} \right],
\end{align*}
so when $t\geq \frac{\balls^{2-1/c}}{\bins}$ the result follows immediately from~\eqref{azb1}. If on the other hand $t< \frac{\balls^{2-1/c}}{\bins}$ then $\frac{t^2}{\balls^{2-1/c}}< \frac{\balls^{2-1/c}}{\bins^2}\leq 1$ and the result is trivial as the right hans size in~\eqref{cbu} can be as large as $\Omega(1)$ which is a valid upper bound on any probability.

We now turn to the proof of~\eqref{cbl}. Letting $\mathcal{E}$ denote the event that $\fbins\leq \mu_0-2t$ and $\mathcal{A}$ the event that $H\leq \trm-t$ we have that
\begin{align*}
\Pr[\fbins\leq \mu_0-2t]=\Pr[\mathcal{E}]\leq \Pr[\mathcal{A}]+\Pr[\mathcal{E} \wedge \neg \mathcal{A}].
\end{align*}
By~\eqref{azb2} and since $\trm \leq \E[H]$ we can upper bound $\Pr[\mathcal{A}]\leq \exp \left(\frac{-t^2}{2\balls^{2-1/c}} \right)$. For the other term we note that $\mathcal{E} \wedge \neg \mathcal{A}$ entails that at least $t$ bins are hit by a dummy ball. In particular the number of dummy balls is at least $t$. As the number $C$ of internal collisions of the groups is an upper bound on the number of dummy balls this in turn implies, $t \leq C$. We may assume that $t\geq \balls^{1-1/(2c)}$ as otherwise~\eqref{cbu} is trivial. As we assumed, $\balls^{1-1/(2c)}\leq \bins$ it follows from Lemma~\ref{varlem} that  $\E[C]\leq \frac{\balls^{2-1/c}}{2\bins}\leq \frac{t \balls^{1-1/(2c)}}{2\bins}\leq t/2$ and so $t-\E[C]\geq t/2$. Applying Chebychev's inequality as well as~\eqref{varcol} of Lemma~\ref{varlem} we thus obtain, 
\begin{align*}
\Pr[\mathcal{E} \wedge \neg \mathcal{A}]\leq \Pr[C\geq t]\leq \Pr[C-\E[C]\geq t/2]\leq \frac{4\var[C]}{t^2}=\OO \left( \frac{\balls^2}{t^2\bins} \right).
\end{align*}
Combining the two bounds completes the proof.
\end{proof}
We promised to argue why we cannot dispose with the term $\balls^2/(\bins t^2)$ in general. Suppose that $\balls=O(\bins)$ and let  $t=\balls^{1/2+\alpha}$ for an $\alpha\in [1/2,1)$ such that $\balls^{1/2+\alpha}\in[ \sqrt{\balls},\balls/2]$, and consider the set of keys $[\balls /t]\times [t]$. With probability $\Omega((\balls/t)^2/\bins)=\Omega(\balls^2/(t^2\bins))$ we have that $h_0(a_0)=h_0(a_1)$ for two distinct $a_0,a_1\in [\balls/t]$. Conditioned on this event the expected number of non-empty bins is at most $\bins\left(1-\left(1- \frac{\balls/t-1}{\bins}\right)^t\right)$ which can be shown to be $\trm-\Omega(t)$ by standard calculus. The additive term $\Omega(t)$ comes from the fact that the $t$ pairs of colliding keys $\{(a_0,b),(a_0,b)\}_{b \in [t]}$ causes the expected number of non-empty bins to decrease by $\Omega(t)$ when $\balls=\OO(\bins)$. Thus the deviation by $\Omega(t)$ from $\trm$ occurs with probability  $\Omega(\balls^2/(\bins t^2))$.
\\

We will now set the stage for the proof of Theorem~\ref{concthm2}. As mentioned in the introduction we require a stronger martingale inequality than that by Azuma. The one we use is due to Mcdiarmid~\cite{Colin}. Again assume that $(\Omega,\mathcal{F},\Pr)$ is a finite probability space, that $X:\Omega \to \R$ is an $\mathcal{F}$-measurable random variable, that $(\mathcal{F}_i)_{i=0}^k$ is a filter of $\mathcal{F}$, and that $X_i=\E[X \mid \mathcal{F}_i]$. Also recall the definition of conditional variance: If $\mathcal{G}\subseteq \mathcal{F}$ is a $\sigma$-algebra, then $\var[X \mid \mathcal{G}]=\E[(X-\E[X\mid \mathcal{G}])^2 \mid \mathcal{G}]=\E[X^2 \mid \mathcal{G}]-\E[X\mid \mathcal{G}]^2$.
\begin{thm}[Mcdiarmid~\cite{Colin}]\label{mcdineq}
Assume that $\var[X_i \mid \mathcal{F}_{i-1}]\leq \sigma_i^2$ for $i=1,\dots,k$ and further that $X_i-X_{i-1}\leq M$ for $i=1,\dots,k$. Then for $t\geq 0$,
$$
\Pr[X-\E[X]\geq t]\leq \exp\left(\frac{-t^2}{2\left(\sum_{i=1}^k \sigma_i^2+Mt/3\right)} \right).
$$
\end{thm}
With this tool in hand we are ready to prove Theorem~\ref{concthm2}, the main technical challenge being to argue why we can apply Theorem~\ref{mcdineq}.
\begin{proof}[\textbf{Proof of Theorem~\ref{concthm2}}]
We introduce dummy balls exactly in the proof of Theorem~\ref{concthm} and Theorem~\ref{expthm} and consider the same martingale $(H_i)_{i=0}^k=(\E[H|\mathcal{F}_i])_{i=0}^k$, where $H$ is the number of bins hit (by either a dummy ball or a ball from $X$). We already saw that $|H_i-H_{i-1}|\leq |G_i|\leq \balls^{1-1/c}$, so we let $M:=\balls^{1-1/c}$. What remains is to upper bound $\var[H_i \mid \mathcal{F}_{i-1}]$. First note that
\begin{align*}
\var[H_i\mid \mathcal{F}_{i-1}]=&\E[(H_i-\E[H_i \mid \mathcal{F}_{i-1}])^2 \mid \mathcal{F}_{i-1}] \\
=&\E[(H_i-H_{i-1})^2 \mid \mathcal{F}_{i-1}].
\end{align*}
We denote by $N_i$ the number of bins that are empty after the hash values of the first $i$ position characters has been settled. Then, by the same reasoning as in the proof of Theorem~\ref{hpthm}, we have that
\begin{align*}
|H_i-H_{i-1}|\leq   \left|N_i-N_{i-1}+\frac{ |G_i| \cdot N_{i-1}}{\bins} \right|.
\end{align*}
Now let $T_i=|G_i|-N_{i-1}+N_i$ denote the number of bins hit in the $i$'th step that were already hit in the $(i-1)$'st step. As $\E[T_i \mid \mathcal{F}_{i-1}\,]=(\bins-N_{i-1})|G_i|/\bins$, the above inequality reads
$$
|H_i-H_{i-1}| \leq |T_i-\E[T_i \mid \mathcal{F}_{i-1}]|,
$$
and so, 
$$
\var[H_i \mid \mathcal{F}_{i-1}]\leq \var[T_i \mid \mathcal{F}_{i-1}]\leq \E[T_i^2 \mid \mathcal{F}_{i-1}].
$$
Now, $T_i^2$ counts the number of 2-tuples $(y,z)$ with $y,z\in\{h(x\backslash \{\alpha_i\}): x \in G_i\}\cup D_i$ such that $h(y)$ and $h(z)$ are already hit after the $(i-1)$'st step. Conditioned on $\mathcal{F}_{i-1}$ the probability that this occurs for a given such pair is at most $\frac{\bins-N_{i-1}}{\bins}\leq \frac{\balls}{\bins}$, and there are exactly $|G_i|^2$ such pairs.  Hence
$$
\var[H_i \mid \mathcal{F}_{i-1}]\leq \frac{\balls}{\bins}|G_i|^2:=\sigma_i^2.
$$
By Lemma~\ref{varlem}, $\sum_{i=1}^k \sigma_i^2\leq \frac{\balls^{3-1/c}}{\bins}$ so Theorem~\ref{mcdineq} gives that for $t\geq 0$
\begin{align}\label{mcdb1}
&\Pr[H-\E[H]\geq t]\leq \exp\left(\frac{-t^2}{2\left(\frac{\balls^{3-1/c}}{\bins}+\balls^{1-1/c}t/3\right)} \right) \\
\leq& \exp \left(-\min \left\{ \frac{t^2}{4\frac{\balls^{3-1/c}}{\bins}},\frac{3t}{2\balls^{1-1/c}} \right\} \right). \nonumber
\end{align}
As $\E[H]\leq \trm+\frac{\balls^{2-1/c}}{\bins}$ this is also an upper bound on $\Pr\left[\fbins \geq \trm+t+\frac{\balls^{2-1/c}}{\bins}\right]$. Now the same argument as in the proof of Theorem~\ref{concthm} leads to the upper bound~\eqref{cbu2}.
\\
\\
Finally, to prove~\eqref{cbl2} we use the same strategy as above but this time we define $H'=-H$ and the martingale $(H'_i)_{i=0}^k=(\E[H' \mid \mathcal{F}_{i} ])_{i=0}^k$. Then $|H_i'-H_{i-1}'|=|H_i-H_{i-1}|\leq M$ and $\var[H_i' \mid \mathcal{F}_{i-1}]=\var[H_i \mid \mathcal{F}_{i-1}]$ for $i=1,\dots,k$, so we get a bound as in~\eqref{mcdb1}, but this time on $\Pr[H'-\E[H']\geq t]=\Pr[H-\E[H]\leq -t]$. 

As in the proof of Theorem~\ref{concthm}, the event $\fbins \leq \trm -t$ implies that either $\mathcal{A}$: $H-\E[H]\leq -t/2$ or $\mathcal{B}$: the number of internal collisions $C$ is at least $t/2$. 
$\Pr[\mathcal{A}]$ is bounded using~\eqref{mcdb1}, giving us the first term of the bound in~\eqref{cbl2}. 
For $\Pr[\mathcal{B}]$, note that we may assume that $t\geq 4\balls^{1-1/c}$ as otherwise~\eqref{mcdb1} is trivial. In that case $\E[C]\leq \balls^{2-1/c}/\bins\leq \frac{t}{4}\frac{\balls}{\bins}\leq \frac{t}{4}$, so $t/2-\E[C]\geq t/4$. Lemma~\ref{varlem} thus gives that $\Pr[\mathcal{B}]=\OO(\frac{\balls^2}{\bins t^2})$ --- the second term in the bound~\eqref{cbl2}. The proof is complete.
\end{proof}

\subsection{Proofs of technical lemmas}\label{techlemsec}
For proving Lemma~\ref{varlem} and Lemma~\ref{ineqlemma} we need to briefly discuss the independence of simple tabulation. In the notion of $k$-independence introduced by Wegman and Carter~\cite{WC} simple tabulation is only $3$-independent as shown by the set of keys $S=\{(a_0,b_0),(a_0,b_1),(a_1,b_0)(a_1,b_1)\}$. Indeed $\bigoplus_{x\in S}h(x)=0$ showing that the keys do not hash independently. The issue is that since each position character appears an even number of times in $S$ the addition over $\Z_2$ causes the terms to cancel out. This property in a sense characterises dependencies of keys as shown by Thorup and Zhang~\cite{Zhang}

\begin{lemma}[Thorup and Zhang~\cite{Zhang}]\label{depchar}
The keys $x_1,\dots,x_k\in U$ are dependent if and only if there exists a non-empty subset $I\subseteq \{1,\dots,k\}$ such that each position character in $(x_i)_{i\in I}$ appears an even number of times. In this case we have that $\bigoplus_{i\in I}h(x_i)=0$.
\end{lemma}

For keys $x,y\in U$ we write $x\oplus y$ for the symmetric difference of $x$ and $y$ when viewed as sets of position characters. Then the property that each position character appearing an even number of times in $(x_i)_{i\in I}$ can be written as  $\bigoplus_{i\in I}x_i=\emptyset$. As shown by Dahlgaard et al.~\cite{Dahl2} we can efficiently bound the number of such tuples $(x_i)_{i\in I}$.
\begin{lemma}[Dahlgaard et al.~\cite{Dahl2}]\label{deplemma} Let $A_1,\dots,A_{2t}\subseteq U$. The number of $2t$-tuples $(x_1,\dots,x_{2t})\in A_1\times \cdots \times A_{2t}$ such that $x_1\oplus \cdots \oplus x_{2t}=\emptyset$ is at most $((2t-1)!!)^c\prod_{i=1}^{2t}\sqrt{|A_i|}$. Here $a!!$ denotes the product of all the positive integers in $\{1,\dots,a\}$ having the same parity as $a$. 
\end{lemma}

We now provide the proofs of Lemma~\ref{varlem} and Lemma~\ref{ineqlemma}. Since we need Lemma~\ref{ineqlemma} in the proof of Lemma~\ref{varlem} we prove that first.
\begin{proof}[\textbf{Proof of Lemma~\ref{ineqlemma}}]
We prove the following more general statement: Let $f:[0,\balls_0]\to \R$ be convex with $f(0)=0$. Let $0\leq g_1,\dots,g_k\leq \balls_0$ be such that $m=\sum_{i=1}^kg_i$. Define $S:=\sum_{i=1}^k f(g_i)$. Then  $S\leq(\balls/\balls_0)f(\balls_0)$.

To see why the statement holds note that by convexity, $f(x)+f(y)\leq f(x-t)+f(y+t)$ if $0\leq t \leq x \leq y\leq \balls_0-t$. To maximize $S$ we thus have to set $k=\lceil\balls/ \balls_0 \rceil$, $g_1=\cdots=g_{k-1}=\balls_0$ and $g_k=\balls-\sum_{i=1}^{k-1}g_i=\eps \balls_0$, where $\eps\in[0,1)$. It follows that
$$
S\leq\left( \frac{\balls}{\balls_0}-\eps \right) f(\balls_0)+ f \left(\eps \balls_0 \right). 
$$
Finally $f(\eps\balls_0)\leq \eps f(\balls_0)$ using convexity and that $f(0)=0$, so $S\leq(\balls/\balls_0)f(\balls_0)$ as desired.

The first inequality~\eqref{simineq1} of the lemma follows immediately from the above statement with $f(x)=x^{\alpha}$ which is convex since $\alpha\geq 1$. 
For inequality~\eqref{simineq2} we may assume that $\bins>\balls_0$ as the result is trivial when $\bins=\balls_0$. We then define $f(x)=-\log(1-x/n)$ which is convex with $f(0)=0$. Then 
$$
S=-\sum_{i=1}^k \log\left(1-\frac{g_i}{n}\right)\leq -\frac{\balls}{\balls_0}\log\left(1-\frac{\balls}{\bins}\right),
$$ 
which upon exponentiation leads to inequality~\eqref{simineq2}.
\end{proof}

\begin{proof}[\textbf{Proof of Lemma~\ref{varlem}}]
We define $g_i=|X_i|$ for $i\in [k]$. Now~\eqref{expcol} is easily checked. Indeed, since simple tabulation is $2$-independent,
\begin{align*}
\E [C]=\sum_{i=1}^k \binom{g_i}{2}\frac{1}{\bins}\leq \frac{1}{2\bins}\sum_{i=1}^t g_i^2\leq \frac{\balls \cdot \balls_0}{2\bins},
\end{align*}
where in the last step we used Lemma~\ref{ineqlemma}. The last statement of the lemma concerning $\E[C \mid h(q)=z]$ follows from the same argument this time however using that simple tabulation is $3$-independent.
%as $g_i\leq 2m^{1-1/c}$ for $i=1,\dots,k$ and the sum of the squares is maximised when $k=\frac{1}{2}m^{1/c}$ and all $g_i=2m^{1-1/c}$. 

We now turn to~\eqref{varcol}. Writing $\var[C]=\E[C^2]-(\E[C])^2$ our aim is to bound $\E[C^2]$. 
Note that $C^2$ counts the number of tuples $(\{x,y\},\{z,w\})$ such that $x\neq y$ and $z\neq w$ but $h(x)=h(y)$ and $h(z)=h(w)$ and furthermore $x,y\in G_i$ and $z,w\in G_j$ for some $i,j\in [k]$. We denote the set of such tuples $T$ and for $\tau=(\{x,y\},\{z,w\})\in T$ we let $X_{\tau}$ be the indicator for the event that both $h(x)=h(y)$ and $h(z)=h(w)$. Then 
\begin{align}\label{expproof1}
\E[C^2]=\sum_{\tau\in T} \Pr(X_{\tau}=1).
\end{align}
We now partition $T$ by letting
\begin{itemize}
\item $T_1$ be the elements of $T$ for which $\{x,y\}=\{z,w\}$.
\item $T_2$ be the elements of $T$  for which $|\{x,y,z,w\}|=3$.
\item $T_3$ be the elements of $T$ for which $x,y,z,w$ are distinct and independent.
\item $T_4$ be the elements of $T$ for which $x,y,z,w$ are distinct and dependent and there is an $i\in [k]$ such that $x,y,z,w \in G_i$.
\item $T_5$ be the remaining elements of $T$, that is, those element $(\{x,y\},\{z,w\})$ such that $x,y,z,w$ are distinct and dependent and such that $\{x,y\}\subseteq G_i$ and $\{z,w\} \subseteq G_j$ for some distinct $i,j\in [k]$.
\end{itemize}
Putting $S_j=\sum_{\tau\in T_j} \Pr(X_{\tau}=1)$ the sum in~\eqref{expproof1} can be written as $\sum_{j=1}^5 S_j$ and we can efficiently upper bound each of the inner sums as we now show. Clearly,
\begin{align*}
S_1=\sum_{i=1}^k\binom{g_i}{2}\frac{1}{\bins}=\E[C].
\end{align*}
For the second sum we use that simple tabulation is $3$-independent and that $|\{x,y,z,w\}|=3$ implies that $x,y,z,w$ belongs to the same group $G_i$ for some $i\in [k]$. Hence
\begin{align*}
S_2=\sum_{i=1}^k \binom{g_i}{2}\cdot 2 \cdot  \binom{g_i-2}{1} \frac{1}{\bins^2}\leq \frac{1}{\bins^2}\sum_{i=1}^k g_i^3\leq \frac{\balls \cdot \balls_0^2}{\bins^2},
\end{align*}
again using Lemma~\ref{ineqlemma} to bound the sum of cubes. Finally we upper bound $S_3$ as
\begin{align*}
S_3\leq\frac{1}{\bins^2} \left(\sum_{i=1}^k \binom{g_i}{2}\binom{g_i-2}{2} +\sum_{i,j \in [k]: i\neq j}\binom{g_i}{2}\binom{g_j}{2}\right)\leq \frac{1}{\bins^2} \left( \sum_{i=1}^k \binom{g_i}{2}\right)^2=  \E[C]^2.
\end{align*}
Note that in the first three steps we have not been using anything about simple tabulation except it being $3$-independent. However, if $(\{x,y\},\{z,w\})\in T_4\cup T_5$ then by Lemma~\ref{depchar} we have that $x\oplus y\oplus z \oplus w=\emptyset$ and thus that $h(x)=h(y)$ exactly if $h(z)=h(w)$ which happens with probability $\bins^{-1}$. Thus in this case we have to efficiently bound the sizes of $T_4$ and $T_5$. Luckily Lemma~\ref{deplemma} comes to our rescue and we can bound
\begin{align*}
S_4+S_5\leq \frac{3^c}{\bins} \left(\sum_{i=1}^kg_i^2+\sum_{i,j\in [k]: i\neq j}g_ig_j\right)=\frac{3^c}{\bins}\left( \sum_{i=1}^k g_i\right)^2=\frac{3^c\balls^2}{\bins}.
\end{align*}
Combining all this we find that 
\begin{align*}
\var[C]=\E[C^2]-\E[C]^2\leq \E [C]+\frac{\balls \cdot \balls_0^2}{\bins^2}+\frac{3^c\balls^2}{\bins}\leq \frac{\balls \cdot \balls_0^2}{\bins^2}+\frac{(3^c+1)\balls^2}{\bins}, 
\end{align*}
as desired.
\end{proof}

\section{Handling bins consisting of many subbins}\label{Projsec}
In this section we show how to modify the proof of Theorem~\ref{expthm} to obtain Theorem~\ref{expthm2}.
\begin{proof}[Proof of Theorem~\ref{expthm2}]
We may assume that $\rho\balls^{1-1/c}\leq 1$ as otherwise the result is trivial.

As usual we consider the ordering on the position characters, $\alpha_1\prec \cdots \prec \alpha_k$, obtained from Lemma~\ref{orderlemma}, and we fix the values $h(\alpha_i)$ in this order. Suppose that $(h(\alpha_j))_{j<i}$ are fixed and let $V_i=\{y\in [2^r] \mid \exists x \in G_i:h(x \backslash \{\alpha_i\})+y \in S \}$ denote those hash values $h(\alpha_i)$ that would cause $S\cap h(G_i)\neq \emptyset$. Note that $V_i$ is a random variabel depending only on $(h(\alpha_j))_{j<i}$. Let $D_i\subseteq [2^r]\backslash V_i$ be a set of \emph{dummy hash values} chosen dependently on $(h(\alpha_j))_{j<i}$ such that $(|D_i|+|V_i|)/2^r=\rho |G_i|$. As $|G_i|\leq \balls^{1-1/c}$ and so $\rho |G_i|\leq  \rho \balls^{1-1/c} \leq 1$ this is in fact possible. We say that $S$ is \emph{hit} if there exists and $i\in\{1,\dots,k\}$ such that $h(\alpha_i)\in V_i\cup D_i$, and we denote this event $\mathcal{H}$. Defining $\balls_0=\balls^{1-1/c}$ we then have
$$
 \Pr[\mathcal{H}]=1-\prod_{i=1}^k (1-|G_i|\rho)\leq 1-\left(1-\balls_0 \rho\right)^{\balls/\balls_0}\leq 1-e^{-\balls \rho}(1-\balls \balls_0\rho^2)\leq \trp' +\balls^{2-1/c} \rho^2,
$$
using the same inequality as in the proof of Theorem~\ref{expthm}. This is clearly also an upper bound on $p=\Pr[h(X)\cap S\neq \emptyset]=\Pr[\bigcup_{i=1}^k (h(\alpha_i)\in V_i)]$.

Now for the lower bound: For $i\in \{1,\dots,k\}$ we let $\mathcal{D}_i$ and $\mathcal{R}_i$ denote events that $h(\alpha_i) \in D_i$  and that $h(\alpha_i) \in V_i$ respectively. Then 
$$
\Pr\left[\bigcup_{i=1}^k \mathcal{D}_i \right]\leq \sum_{i=1}^k\Pr[\mathcal{D}_i]=\sum_{i=1}^k (\Pr[\mathcal{R}_i\cup \mathcal{D}_i]-\Pr[\mathcal{R}_i])= \rho \balls-\sum_{i=1}^k \Pr[\mathcal{R}_i].
$$
By the Bonferroni inequality, and $2$-independence
$$
\Pr[\mathcal{R}_i]=\Pr[h(G_i)\cap S \neq \emptyset]\geq |G_i| \rho-\binom{|G_i|}{2}\rho^2,
$$
so it follows that $\Pr[\bigcup_{i=1}^k \mathcal{D}_i]\leq \sum_{i=1}^k \binom{|G_i|}{2}\rho^2 \leq \balls^{2-1/c}\rho^2$. Finally
$$
p\geq \Pr[\mathcal{H}]-\Pr\left[\bigcup_{i=1}^k \mathcal{D}_i\right] \geq\trp'-\balls^{2-1/c}\rho^2,
$$
which completes the proof of the lower bound. 

The case where we condition on the event $\mathcal{E}$ that $h(q)=z$ for a $z\in [2^r]$ is handled analogously but this time choosing the order $\prec$ as described in the second part of Lemma~\ref{orderlemma}. The upper bound on $p$ then follows as before and for the lower bound we use $3$-independence of simple tabulation when applying the Bonferroni inequality to lower bound $\Pr[\mathcal{V}_i \mid \mathcal{E}]$.\end{proof}

\newpage

\bibliographystyle{plain}
\bibliography{bibliography}

\begin{comment}
\newpage
We want to bound 
\begin{align*}
&\var[H_i \mid \mathcal{F}_{i-1}]=\E[(H_i-\E[H_i \mid \mathcal{F}_{i-1}])^2 \mid \mathcal{F}_{i-1}]= \E[(H_i-H_{i-1})^2 \mid \mathcal{F}_{i-1}] \\
=& \E[(T_i-T_{i-1})^2 \mid \mathcal{F}_{i-1}]
\end{align*}
Letting $E_i$ be the number of empty bins after first $i$ groups have been thrown, we have $T_i=E_i\prod_{j>i}\left(1-\frac{g_j}{\bins} \right)$. Thus
\begin{align*}
&\var[H_i \mid \mathcal{F}_{i-1}]\leq \E\left[\left(E_i-\left(1-\frac{g_i}{n} \right)E_{i-1}\right)^2 \mid \mathcal{F}_{i-1} \right]=\E\left[\left((E_{i-1}-E_i)-\frac{g_i}{n} E_{i-1}\right)^2 \mid \mathcal{F}_{i-1} \right] \\
= & \var[E_{i-1}-E_i \mid \mathcal{F}_{i-1}]
\end{align*}
But $E_{i-1}-E_i=|G_i|-L_i$ where $L_i$ is the number of balls from $G_i$ that clashes with one of the balls from the earlier groups $(G_j)_{j< i}$. Hence the above is
\begin{align*}
\var[|G_i|-L_i \mid \mathcal{F}_{i-1}]=\var[L_i \mid \mathcal{F}_{i-1}]\leq \E[L_i^2 \mid \mathcal{F}_{i-1}]\leq g_i^2 \frac{\balls}{\bins}.
\end{align*}
It follows that 
\begin{align*}
\sum_i \var[H_i \mid \mathcal{F}_{i-1}] \leq \frac{\balls^{3-1/c}}{\bins}
\end{align*}
By McDiarmid's inequality 
\begin{align}
\Pr[H-\E[H]\geq t] \leq \exp \left( \frac{-t^2}{2\frac{\balls^{3-1/c}}{\bins}+\frac{2}{3}\balls^{1-1/c}t} \right)
\end{align}?
\end{comment}

\end{document}